\documentclass[twocolumn,superscriptaddress,aps,pre]{revtex4-2}

\usepackage[utf8]{inputenc}
\usepackage{url}
\usepackage{amssymb,amsmath,amsthm}
\usepackage{bm}
\usepackage{graphicx}
\usepackage{enumerate}
\usepackage{microtype}
\usepackage{float}
\usepackage{color}
\usepackage{hyperref}
\usepackage{multirow}
\usepackage{array}
\newcolumntype{C}[1]{>{\centering\let\newline\\\arraybackslash\hspace{0pt}}m{#1}}

\usepackage[ruled,linesnumbered,vlined]{algorithm2e}

\newtheorem*{theorem}{$\mathbf{Theorem}$}

\begin{document}

\title{Rigidity control of general origami structures}
 
\author{Rongxuan Li}
\affiliation{Department of Mathematics, The Chinese University of Hong Kong}
\author{Gary P. T. Choi}
\thanks{To whom correspondence may be addressed. Email: ptchoi@cuhk.edu.hk.}
\affiliation{Department of Mathematics, The Chinese University of Hong Kong}

\begin{abstract}
Origami, the traditional paper-folding art, has inspired the modern design of numerous flexible structures in science and engineering. In particular, origami structures with different physical properties have been studied and utilized for various applications. More recently, several deterministic and stochastic approaches have been developed for controlling the rigidity or softness of the Miura-ori structures. However, the rigidity control of other origami structures is much less understood. In this work, we study the rigidity control of general origami structures via enforcing or relaxing the planarity condition of their polygonal facets. Specifically, by performing numerical simulations on a large variety of origami structures with different facet selection rules, we systematically analyze how the geometry and topology of different origami structures affect their degrees of freedom (DOF). We also propose a hypergeometric model based on the selection process to derive theoretical bounds for the probabilistic properties of the rigidity change, which allows us to identify key origami structural variables that theoretically govern the DOF evolution and thereby the critical rigidity percolation transition in general origami structures. Moreover, we develop a simple unified model that describes the relationship between the critical percolation density, the origami facet geometry, and the facet selection rules, which enables efficient prediction of the critical transition density for high-resolution origami structures. Altogether, our work highlights the intricate similarities and differences in the rigidity control of general origami structures, shedding light on the design of flexible mechanical metamaterials for practical applications.

\end{abstract}

\maketitle

\section{Introduction}
Origami (paper folding) has a long history in various cultures~\cite{hatori2011history} and was commonly used for ceremonial and recreational purposes. Over the past several decades, it has become increasingly popular among not just artists but also scientists and engineers, and numerous efforts have been devoted to the creation and analysis of different origami structures~\cite{huffman1976curvature,hull1994mathematics,lang1996computational,kawasaki2005roses,demaine2007geometric,tachi2009origamizing,lang2012origami} as well as their applications to the design of soft robots~\cite{rus2018design,ze2022soft}, logic gates~\cite{treml2018origami,meng2021bistability}, and aerospace structures~\cite{nishiyama2012miura}.

\begin{figure}[t!]
    \centering
    \includegraphics[width=0.97\linewidth]{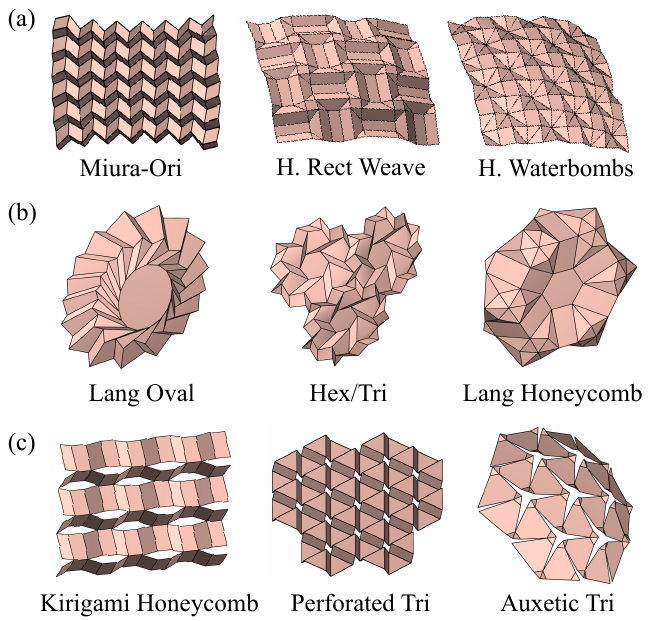}
    \caption{\textbf{Different classes of origami structures considered in our study.} (a)~Periodic origami structures including the Miura-ori, the Huffman Rectangular Weave, and the Huffman Waterbombs. (b)~Rotational origami structures including the Lang Oval, Hex/Tri, and Lang Honeycomb. (c)~Perforated origami structures including the Kirigami Honeycomb, Perforated Triangle, and Auxetic Triangle.}
    \label{fig:manyorigami}
\end{figure}

Miura-ori structures~\cite{miura1985method}, as a prime example of origami structures with widespread applications in science and engineering, have been extensively studied. In particular, several prior works have explored their mechanical properties~\cite{wei2013geometric,schenk2013geometry} and geometric design~\cite{dudte2016programming,dudte2021additive}. In a recent work~\cite{chen2019rigidity}, Chen and Mahadevan studied the stochastic control of the rigidity of Miura-ori structures. More recently, Li and Choi~\cite{li2025explosive} studied the rigidity percolation transition in floppy Miura-ori structures using the idea of explosive percolation~\cite{achlioptas2009explosive,araujo2010explosive,riordan2011explosive}. Besides Miura-ori structures, the rigidity of some other origami structures has also been analyzed in recent studies~\cite{chen2018branches,he2019rigid,he2020rigid,he2022rigid,zhang2023rigidity}. However, the rigidity control of more general origami structures remains less understood. Specifically, how does the rigidity or floppiness of general origami structures change if one enforces or relaxes the planarity property of individual facets? Is it possible to achieve explosive rigidity percolation in general origami structures via carefully designed selection rules? How does the rigidity percolation transition vary with the geometric and topological properties of different origami structures? In this work, we address these questions by performing both experimental and theoretical analyses on various origami structures.

Specifically, here we consider three major classes of origami structures commonly used in practical applications (Fig.~\ref{fig:manyorigami}), namely (a) the \emph{Periodic Origami}, (b) the \emph{Rotational Origami}, and (c) the \emph{Perforated Origami}. In the class of \emph{Periodic Origami} structures, the folds are formed in a periodic and scalable manner. One classical example is the Miura-ori pattern~\cite{miura1985method}, which consists of identical four-coordinated quadrilateral facets. Two other examples are the Huffman Rectangular Weave pattern, consisting of triangular, rectangular, and trapezoidal facets, and the Huffman Waterbombs, consisting of triangular and square facets, both by David Huffman~\cite{davis2013reconstructing}. In the class of \emph{Rotational Origami} structures, the creases form foldable structures with rotational symmetry. Examples include the Lang Oval (by Robert Lang), the Hex/Tri tessellation (by Kendrick Feller), and the Lang Honeycomb (by Robert Lang)~\cite{lang2012origami}. For the \emph{Perforated Origami} structures, the origami folds are designed on a perforated sheet, with examples including the Kirigami Honeycomb, the Perforated Triangle by Johann Kreuter, and the Auxetic Triangle tessellations (see also~\cite{origamisimulator,ghassaei2018fast} and Appendix~\ref{appendix:origami} for more details). For each of these representative patterns in the three classes of origami structures, we study how enforcing or relaxing the planarity of individual facets in the structures can lead to a change in their overall degrees of freedom (DOF). We further analyze and compare the DOF evolution, the critical transitions, and the transition sharpness of different structures to understand their underlying similarities and differences. Moreover, the prediction of critical transitions in complex systems has long been of interest~\cite{scheffer2009early, scheffer2012anticipating, kong2021machine}, and the critical transition density in origami structures also plays a key role in guiding applications in mechanical memory~\cite{treml2018origami,han2023origami}, tunable stiffness and deployable structures~\cite{zang2024kresling}, and adaptive metamaterials~\cite{chen2019rigidity}. Operating near the critical transition density enables programmable transitions between floppy and rigid states, which is essential for reconfigurable design and efficient mechanical information encoding in origami-based systems. We are therefore motivated to derive simple formulas that relate the rigidity percolation critical transition to origami structural parameters and facet selection rules. This provides an effective framework for predicting the critical point of rigidity gain or loss in various origami structures.

\begin{figure}[t]
    \centering
    \includegraphics[width=\linewidth]{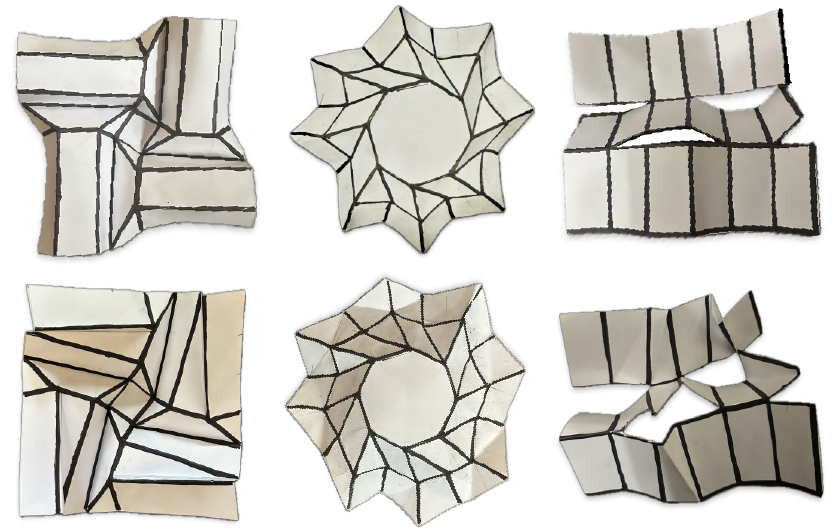}
    \caption{\textbf{Paper-folded origami models with different folding motions.} The top row shows photographs of a folded configuration of the Huffman Rectangular Weave, Lang Oval, and Kirigami Honeycomb structures. The bottom row shows an alternative folding motion of each structure achieved by relaxing the planarity condition of certain facets.} 
    \label{fig:paper_model}
\end{figure}

\section{Methods}
As described in prior works~\cite{chen2019rigidity,li2025explosive}, the Miura-ori structure is highly floppy if we allow all its quadrilateral facets to bend along the facet diagonals, while it is 1-DOF if all quadrilateral facets are enforced to be planar. Also, in between these two maximally floppy and maximally rigid states, one can enforce the planarity of certain facets sequentially based on different rules to control the rigidity transition behaviors of Miura-ori. It is natural to ask whether one can control the rigidity of more general origami structures in a similar manner and whether the transition behaviors depend on the geometry and topology of the origami structures. As demonstrated by the physical paper models in Fig.~\ref{fig:paper_model} (see also Supplementary Videos 1--2), different folding motions can be achieved by relaxing the planarity of certain facets in different origami structures. Therefore, here we consider a general origami structure and start from a maximally floppy initial state in which all of its facets are allowed to bend. We then study how the rigidity of the structure evolves from the initially floppy state to the maximally rigid state under different selection rules for enforcing the facet planarity (Fig.~\ref{fig:rigidity_control_illustration}(a)).

\begin{figure}[t]
    \centering
    \includegraphics[width=\linewidth]{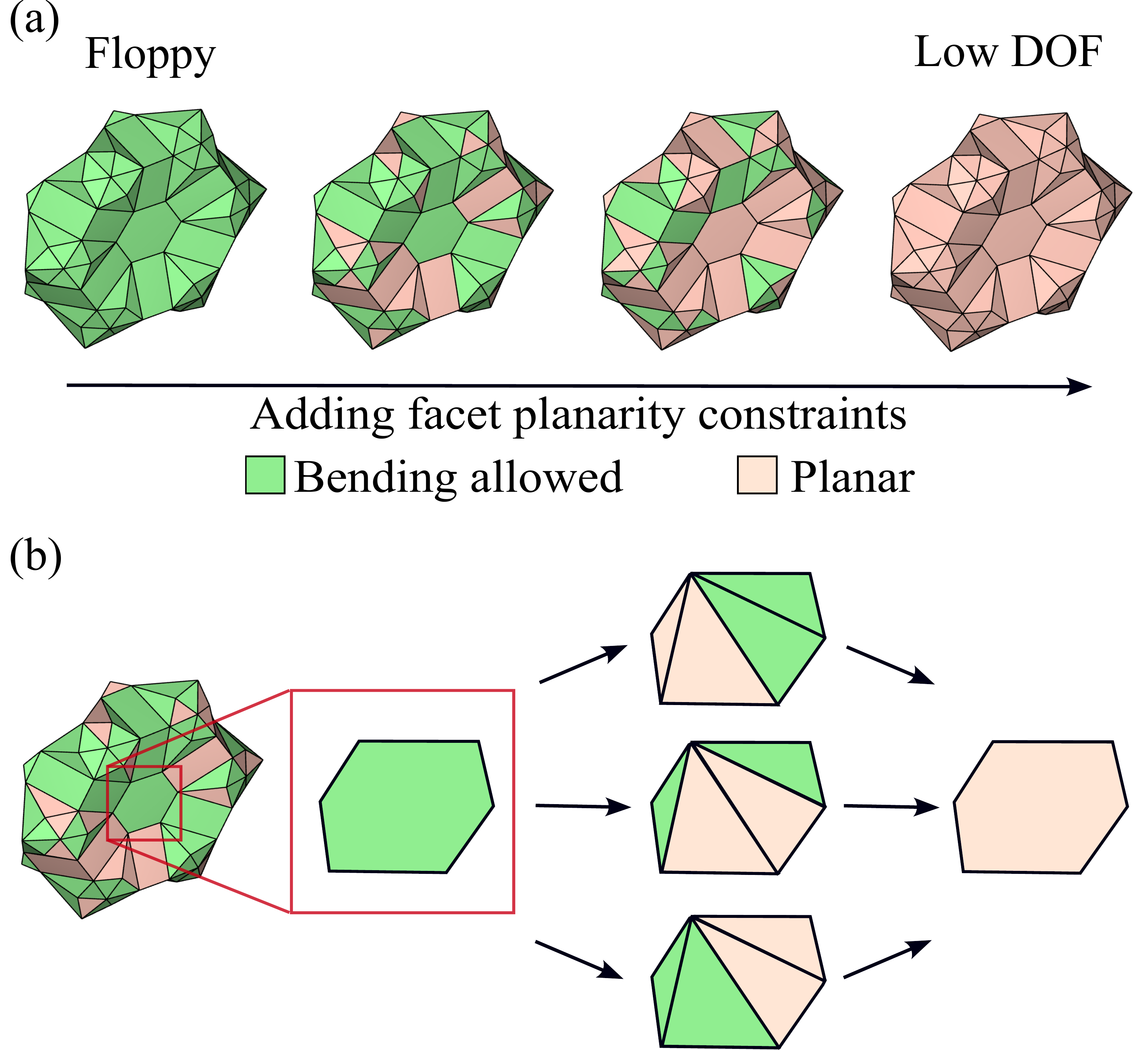}
    \caption{\textbf{An illustration of the rigidity control process.} (a)~Starting from a maximally floppy state of a given origami structure, we consider adding facet planarity constraints based on different selection rules, thereby controlling the rigidity of the structure. (b) To enforce the planarity of a general $n$-sided polygonal facet, we first triangulate it by considering $n-3$ diagonal edges. Then, by enforcing the planarity for every tetrahedron formed by a pair of triangles in the polygonal facet, we can ensure the planarity of the entire facet.}
\label{fig:rigidity_control_illustration}
\end{figure}

\subsection{Assessing the rigidity or softness}
To assess the rigidity or softness of an origami structure, one has to formulate its geometrical constraints and determine the range of its infinitesimal modes of motion. Note that general origami structures may be composed of not only quadrilateral but also triangular, hexagonal, or other polygonal facets. Therefore, here we generalize the approach in~\cite{chen2019rigidity,li2025explosive} to systematically formulate the \emph{edge constraint}, \emph{no-shear constraint}, and \emph{facet planarity constraint} for origami structures with arbitrary polygonal facets.

Specifically, the \emph{edge constraint} for a pair of vertices, corresponding to an edge in the origami structure, enforces that the edge length remains fixed. For each edge $(\mathbf{v}_i, \mathbf{v}_j)$, the constraint is given by 
\begin{equation} \label{eqt:edge_length_constraint}
    g_e = \|\mathbf{v}_i - \mathbf{v}_j\|^2 - l_{ij}^2 = 0,
\end{equation}
where $l_{ij}$ denotes the prescribed length of the edge between vertices $\mathbf{v}_i$ and $\mathbf{v}_j$.

Next, the \emph{no-shear constraint} for every facet in the given origami structure prevents the facet from shearing. For the case of Miura-ori~\cite{chen2019rigidity}, this constraint can be enforced by simply adding a diagonal edge in every quadrilateral facet and imposing an additional edge constraint on it. In our case of general origami structures, we consider triangulating each polygonal facet and imposing edge constraints on all internal edges of the triangulation. More specifically, if we triangulate a $n$-sided polygonal facet ($\mathbf{v}_{1}, \mathbf{v}_{2},\dots, \mathbf{v}_{n}$) (with $n > 3$) by introducing $n-3$ internal edges $(\mathbf{v}_{i_1}, \mathbf{v}_{j_1}), \dots, (\mathbf{v}_{i_{n-3}}, \mathbf{v}_{j_{n-3}})$, the no-shear constraint for this polygonal facet is then given by the edge constraints on these internal edges:   
\begin{equation} \label{eqt:edge_length_constraint_internal}
    g_e = \|\mathbf{v}_{i_p} - \mathbf{v}_{j_p}\|^2 - l_{i_p j_p}^2 = 0,
\end{equation}
where $p = 1,\dots,n-3 $ and $l_{i_p j_p}$ denotes the prescribed length of the internal edge between vertex $\mathbf{v}_{i_p}$ and the reference vertex $\mathbf{v}_{j_p}$. By fixing the internal edge lengths of a triangulated polygonal facet, shearing is uniformly prevented. We remark that there are many valid ways to triangulate a polygonal facet, and the specific choice of the triangulation will not affect the outcome (i.e., as long as the triangulation is valid, the shear is effectively constrained). Also, note that for a triangular facet (i.e., $n = 3$), no no-shear constraints are needed. For an origami structure containing $t_{n-3}$ polygonal facets with $n$ vertices, a total of $ (n-3) \cdot t_{n-3}$ internal edge constraints (as no-shear constraints) are added to prevent shear.

We then consider the \emph{facet planarity constraint}, which prevents a facet in the origami structure from bending. As proposed in~\cite{chen2019rigidity}, the planarity constraint of a quadrilateral facet can be imposed by enforcing the volume of the tetrahedron formed by its four vertices to be 0. More specifically, for a quadrilateral facet formed by four vertices $\mathbf{v}_1, \mathbf{v}_2, \mathbf{v}_3, \mathbf{v}_4$, the facet planarity constraint can be formulated using a scalar triple product as follows: 
\begin{equation}
\label{eqt:quad_planarity_scalar}
    g_p = \left[(\mathbf{v}_{2} - \mathbf{v}_{1}) \times (\mathbf{v}_{4} - \mathbf{v}_{1})\right] \cdot (\mathbf{v}_{3} - \mathbf{v}_{1}) = 0.
\end{equation}
Equivalently, the two triangles $(\mathbf{v}_1, \mathbf{v}_2, \mathbf{v}_3)$ and $(\mathbf{v}_1, \mathbf{v}_3, \mathbf{v}_4)$ will lie on the same plane. For our case of general origami structures, consider a polygonal facet with $n$ vertices $\mathbf{v}_1, \mathbf{v}_2, \dots, \mathbf{v}_n$. Now, note that Eq.~\eqref{eqt:quad_planarity_scalar} can serve as a ``sub-planarity'' constraint that enforces the planarity of the first four vertices $\mathbf{v}_1, \mathbf{v}_2, \mathbf{v}_3, \mathbf{v}_4$ but not necessarily the entire polygonal facet. To fully enforce the planarity of the $n$-sided polygonal facet, we need to introduce $n - 3$ sub-planarity constraints on its $n - 2$ triangulated sub-facets, such that each sub-constraint enforces the planarity of the tetrahedron formed by a pair of adjacent triangles within the polygonal facet. Together, these sub-constraints span the entire polygonal facet and ensure its overall planarity (see Fig.~\ref{fig:rigidity_control_illustration}(b) for an illustration). More mathematically, we show that $n - 3$ is the necessary number of constraints required to enforce the planarity for an $n$-sided polygonal facet.

\begin{theorem}\label{thm:necessary_number}
The number of sub-planarity constraints required to control the planarity of an $n$-sided polygonal facet is exactly $n - 3$. 
\end{theorem}
\begin{proof}
See Appendix~\ref{appendix:origami}.
\end{proof}

We can then construct the infinitesimal rigidity matrix $A$ as described in~\cite{guest2006stiffness} to study the possible infinitesimal modes of motion. Suppose the origami structure has $E$ edges and $V$ vertices. Let the structure contain polygonal facets with up to $n$ edges. Denote $N_0$, $N_1$, $N_2$, \ldots, $N_{n-3}$ as the number of triangles, quads, pentagons, up to $n$-sided facets, respectively, and let the total number of facets be 
\begin{equation}
N = N_{0} + N_{1} \cdots +N_{n-3}.
\end{equation}
We impose planarity constraints on these facets, where the number of planarity constraints is $M_0$ for triangles, $M_1$ for quadrilaterals, $M_2$ for pentagons, and up to $M_{n-3}$ for $n$-gons. The total number of planarity constraints is
\begin{equation}
M = M_0 + M_1 + M_2 + \cdots + M_{n-3}.
\end{equation}
Since enforcing planarity on an \( n \)-sided facet requires \( n - 3 \) sub-planarity constraints on its triangulated sub-facets, we define the total number of constraints on the triangulated origami as 
\begin{equation}
K = E + \sum_{i=0}^{n-3} i \cdot (M_i+N_i).    
\end{equation}
The infinitesimal rigidity matrix $ A \in \mathbb{R}^{K \times 3V}$, constructed on the triangulated origami, encodes edge, facet, and planarity constraints. Its rank indicates the number of infinitesimal DOF of the origami structure when specific planarity conditions are imposed on polygonal facets: 
\begin{equation}
    A = \begin{pmatrix}
        \frac{\partial g_1}{\partial x_1} &  \frac{\partial g_1}{\partial y_1}  &  \frac{\partial g_1}{\partial z_1} & \frac{\partial g_1}{\partial x_2} &  \frac{\partial g_1}{\partial y_2}  &  \frac{\partial g_1}{\partial z_2} & \dots & \frac{\partial g_1}{\partial z_{V}}\\
        \frac{\partial g_2}{\partial x_1} &  \frac{\partial g_2}{\partial y_1}  &  \frac{\partial g_2}{\partial z_1} & \frac{\partial g_2}{\partial x_2} &  \frac{\partial g_2}{\partial y_2}  &  \frac{\partial g_2}{\partial z_2} & \dots& \frac{\partial g_2}{\partial z_{V}}\\
        \vdots & \vdots & \vdots & \vdots & \vdots & \vdots & \ddots & \vdots\\
        \frac{\partial g_K}{\partial x_1} &  \frac{\partial g_K}{\partial y_1}  &  \frac{\partial g_K}{\partial z_1} & \frac{\partial g_K}{\partial x_2} &  \frac{\partial g_K}{\partial y_2}  &  \frac{\partial g_K}{\partial z_2} & \dots& \frac{\partial g_K}{\partial z_{V}}
    \end{pmatrix},
\end{equation}
where $g_1, g_2, \dots, g_K$ include all edge constraints$\{g_{e_j}\}_{i=1}^{E}$, all no-shear constraints $\{g_{n_j}\}_{j=1}^{\sum_{i=1}^{n-3} i \cdot N_i}$, all diagonal constraints, and the current set of planarity constraints on triangulated polygonal facets $\{g_{k_j}\}_{k=1}^{\sum_{i=1}^{n-3} i \cdot M_i}$, and $(x_i, y_i, z_i)$ are the coordinates of the vertex $\mathbf{v}_i$, where $i = 1, 2, \cdots, V$.

Now, suppose there is an infinitesimal displacement $\overrightarrow{dv}$ added to all vertex coordinates $\vec{v} = [x_1, y_1, z_1, x_2, y_2, z_2, \dots, x_V, y_V, z_V]^T$. The condition for infinitesimal rigidity is given by 
\begin{equation}
    A \overrightarrow{dv} = 0.
\end{equation}
Thus, the infinitesimal degrees of freedom (DOF) of the origami structure correspond to the dimension of the null space of $A$. Removing the six trivial global rigid motions (three translations and three rotations), the DOF is given by 
\begin{equation}\label{eqt:DOF}
    d = 3E - \text{rank}(A) - 6.
\end{equation}
The infinitesimal rigidity matrix $A$ for the initial maximally floppy structure includes only the $E$ edge constraints and the $\sum_{i=0}^{n-3} i \cdot N_i$ no-shear constraints (i.e., $M = 0$). Therefore, the number of DOF of the initial structure is given by 
\begin{equation}
d_{\text{initial}} = 3V - E - \sum_{i=0}^{n-3} i \cdot N_i.
\end{equation}
The final DOF, denoted $d_{\text{final}}$, depends on the geometry and topology of the origami structure and is non-negative. 

\subsection{Rigidity control and the DOF evolution}
We then study the evolution of DOF from $d_{\text{initial}}$ to $d_{\text{final}}$ as planarity constraints are gradually imposed. More specifically, at each step, we select a new facet, either a triangular facet or a polygonal facet, and explicitly impose its corresponding planarity constraint. Note that the number of planarity constraints depends on the type of polygonal facet selected. If a triangular facet is selected, no planarity constraint needs to be enforced. For a general polygon with $n$ vertices, $n - 3$ planarity constraints will be imposed if the polygon is selected. We can then define the \emph{planarity constraint density} $\rho \in [0, 1]$ as
\begin{equation}
    \rho = \frac{\text{Number of facets selected}}{\text{Total number of facets}},
\end{equation}
and study how the DOF changes as the density $\rho$ increases from 0 to 1.

Moreover, we can follow the idea of explosive percolation~\cite{achlioptas2009explosive,radicchi2009explosive} and consider multiple candidate facets at each step and select one among them based on certain selection rules to control the rigidity percolation transition.
Let $k \geq 1$ be a positive integer. At each step, we sample $k$ facets randomly from the set of all available facets that have not been selected. We then select one among them based on one of the following selection rules:
\begin{itemize}
    \item \emph{Most Efficient selection rule}: Given $k$ randomly sampled candidate facets $f_1, f_2, \dots, f_k$, we temporarily add the facet planarity constraint of each $f_i$ to the current rigidity matrix $A$ to form $A_i$ and compute the resulting DOF $d_i$. Among all $k$ candidate facets, we select the facet that gives the minimum DOF, i.e., the facet $f_c$ with $c = \text{argmin}_i d_i$. If multiple candidates give the minimum DOF, one is selected randomly.
    \item \emph{Least Efficient selection rule}: Analogous to the above rule, for each candidate facet $f_i$, we construct the augmented matrix $A_i$ and compute the DOF $d_i$. We then select the facet that gives the maximum DOF, i.e., $f_c$ with
    $c = \arg\max_i d_i$. If multiple facets give the maximum DOF, we select one among them randomly.
\end{itemize}
We can then study the effect of the power-of-choices strategy for the rigidity control of general origami structures with different geometric and topological properties.

\subsection{Hypergeometric model}
To quantitatively study the relationship between origami geometry and selection rules with the rigidity percolation (critical transition density) of different origami structures, we need to consider both the selection parameters (i.e., the selection rule and the number of choices $k$) and the structural properties of the origami. To identify structural properties that govern the evolution of DOF, here we estimate and bound the theoretical probabilities of the DOF remaining unchanged or decreasing at a given density $\rho$ using a hypergeometric model. In this model, the triangular facet ratio serves as a key parameter, highlighting it as a fundamental structural property that directly influences the evolution of DOF.

Now, denote the triangular facet ratio at constraint density $\rho$ by $t(\rho)$, and let $N(\rho)$ be the total number of available facets for selection. Then, the number of triangular facets is given by $T(\rho) = t(\rho) N(\rho)$. Let $X(\rho)$ be the random variable representing the number of triangular facets among the $k$ candidate facets selected at density $\rho$. The distribution of $X(\rho)$ follows a hypergeometric distribution:
\begin{equation}
P(X(\rho) = n) = \frac{\dbinom{T(\rho)}{n} \dbinom{N(\rho) - T(\rho)}{k - n}}{\dbinom{N(\rho)}{k}}.
\end{equation}

Now, we denote $P_0(\rho)$ as the probability that the number of DOFs remains unchanged when the constraint density increases by one unit step, i.e.,
\begin{equation}
P_0(\rho) = \mathbb{P}\left(d(\rho) = d\left(\rho + \frac{1}{\text{total \# of facets})}\right)\right),
\end{equation}
and denote $P_1(\rho)$ as the probability that the number of DOFs decreases, i.e.,
\begin{equation}
P_1(\rho) = \mathbb{P}\left(d(\rho) > d\left(\rho + \frac{1}{(\text{total \# of facets})}\right)\right).
\end{equation}

Notice that under the Least Efficient selection rule, a sufficient condition for the DOF $d$ to stay the same at $\rho$ is that at least one triangular facet is included among the $k$ candidates at the selection at $\rho$. In this case, it is guaranteed that a facet can be selected without causing a decrease in DOF by the rule. However, this condition is not necessary for the DOF to remain unchanged, as even when no triangular facet is present in the $k$ candidates, some non-triangular redundant facets may exist among the candidates, resulting in $d$ unchanged after the selection.

Now, note that the probability $P_0(\rho)$ that the DOF remains unchanged can be bounded below by the probability that at least one triangular facet exists among the $k$ candidates,
\begin{equation}
P_0(\rho) \geq 1 - P(X(\rho) = 0) \approx 1 - (1 - t(\rho))^k.
\end{equation}
The approximation of the hypergeometric distribution holds when $N \gg k$. Since the DOF is non-increasing throughout the process, we have the total probability split as $P_1(\rho) + P_0(\rho) = 1$, where $P_1(\rho)$ denotes the probability that the DOF decreases. Therefore, we obtain,
\begin{equation}
P_1(\rho) \leq P(X(\rho) = 0) \approx (1 - t(\rho))^k.
\end{equation}

Notice that under the Most Efficient selection rule, a sufficient condition for the DOF $d$ to remain unchanged is that all $k$ selected facets are triangular. However, this condition is not necessary, as even if some non-triangular facets are included in the candidates, they may still be redundant with respect to the previously selected facets and thus may not lead to a DOF decrease after employing the rule.

In this case, the probability $P_0(\rho)$ that the DOF remains unchanged can be bounded below by the probability that all $k$ candidates are triangular facets,
\begin{equation}
P_0(\rho) \geq P(X(\rho) = k) \approx t(\rho)^k.
\end{equation}
Similarly, we obtain,
\begin{equation}
P_1(\rho) \leq 1 - P(X(\rho) = k) \approx 1 - t(\rho)^k.
\end{equation}

The above inequalities indicate that the change in DOF at constraint density $\rho$ is governed by the number of choices $k$ and the triangular facet ratio $t(\rho)$. 

Note that the behavior of the triangular facet ratio $t(\rho)$ is influenced by the initial triangular facet ratio $t = t(0)$, the selection rule $r$, and the number of choices $k$. To isolate structural effects, we record the initial triangular facet ratio $t = t(0)$, which serves as a structural property of the origami and depends solely on the origami type and its resolution. For each origami structure, we record the initial triangular facet ratio $t$, the selection rule $r$, and then perform numerical simulations as described in the following section.

\section{Results}
To study how the selection rules and the number of choices $k$ affect the rigidity control of general origami structures, we performed numerical simulations in \textsc{MATLAB} with the Parallel Computing Toolbox used. The infinitesimal rigidity matrix $A$ was constructed in sparse matrix format. For the DOF calculation, we followed the approach in~\cite{chen2020deterministic}, using the built-in column approximate minimum degree permutation function \texttt{colamd} and the \texttt{qr} decomposition to compute the QR factorization of $A$. The rank of $A$ was then approximated by counting the number of non-zero diagonal entries in the resulting upper triangular matrix $R$. Since the rank approximation of large matrices may be affected by numerical errors, and since all origami structures considered in our work are foldable, we further restrict the computed DOF values to lie within the feasible range $[1, d_{\text{initial}}]$. Also, note that the construction of the rigidity matrix involves the coordinates of the vertices. In our experiments, we considered all types of origami structures in Fig.~\ref{fig:manyorigami} at a folded configuration with a folding percentage of 25\% generated using the method from~\cite{ghassaei2018fast}. In Appendix~\ref{appendix:origami}, we further compare the simulation results with different folding percentages and show that our analyses also hold for other folding percentages.

For each type of origami structure, each choice of $k = 1, 2, 4, 8, 16, 32$, and each selection rule, we performed 100 independent simulations. In each simulation, we start from the maximally floppy state ($\rho = 0$) and select facets based on selection rules until reaching the maximally rigid state ($\rho = 1$). To facilitate comparison across different origami structures, we also consider the normalized DOF $\widetilde{d}$ defined by 
\begin{equation}
    \widetilde{d} = \frac{d-1}{d_{\text{initial}}-1}.
\end{equation}
Here, we remark that the final DOF of different structures $d_{\text{final}}$ may vary and is not always equal to 1. Nevertheless, we normalize the DOF using $d_{\text{initial}} - 1$, instead of $d_{\text{initial}} - d_{\text{final}}$, to allow a more consistent comparison, particularly with Miura-ori. This normalization also helps us observe whether a structure tends toward a single-DOF mechanism or maintains multiple DOFs when all planarity constraints are enforced.

\begin{figure}[t!]
    \centering
    \includegraphics[width=\linewidth]{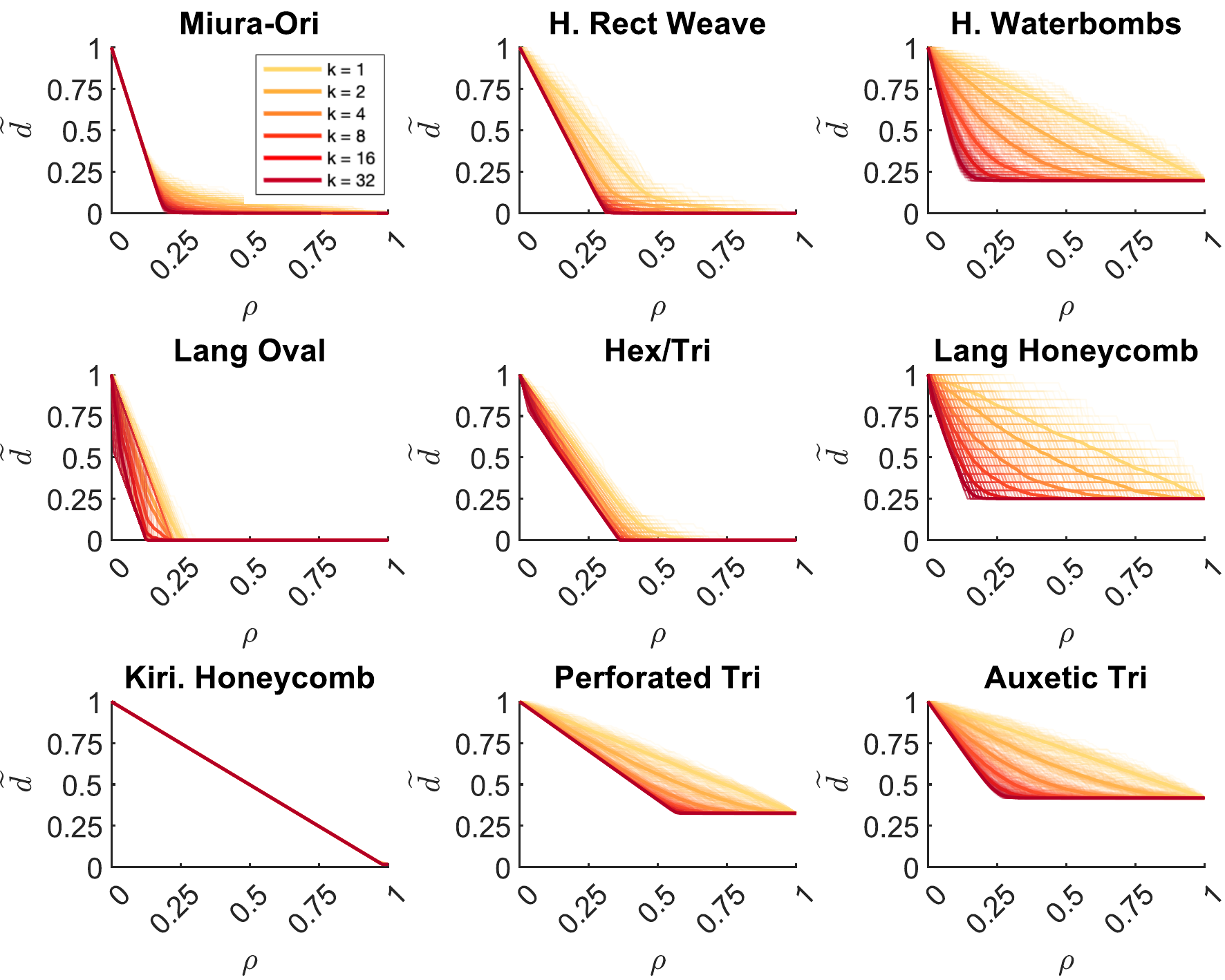}
    \caption{\textbf{Change in the normalized DOF under the Most Efficient selection rule with different numbers of choices for nine types of origami structures.} The three rows correspond to the Periodic Origami, Rotational Origami, and Perforated Origami structures, respectively. Within each row, the ratio of triangular facets increases from left to right across the structures.}
    \label{fig:rule1}
\end{figure}

\subsection{The Most Efficient selection rule}

In Fig.~\ref{fig:rule1}, we plot the values of $\widetilde{d}$ from 100 simulations under the Most Efficient selection rule across different values of $k$ for the 9 origami patterns. The lighter transparent curves represent individual results from each of the 100 simulations for the corresponding $k$ values, while the darker smooth curves show the mean over the 100 simulations for each setup (see Appendix~\ref{appendix:DOF} for more plots with different pattern sizes).


While the Miura-ori structure exhibits a well-defined two-phase evolution of the normalized DOF $\widetilde{d}$, characterized by an initial linear decrease followed by a sharp transition to a sublinear regime across all values of $k$~\cite{chen2019rigidity, li2025explosive}
(see the uniformly linear decline for all $k$ in the Miura-ori results in Fig.~\ref{fig:rule1}), general origami structures show more diverse behaviors. In particular, structures with a high proportion of triangular facets, such as the Huffman Waterbombs, Auxetic Triangle, and Lang Honeycomb, often deviate from this linear trend when $k$ is small (see the yellow regions of the corresponding structures where the DOF shows a scattered decline in Fig.~\ref{fig:rule1}). This difference arises because imposing planarity on triangular facets does not reduce the global DOF, especially during early stages when such facets are inevitably selected from the small $k$-candidate pool. As a result, the early-stage linear decrease observed in Miura-ori becomes less consistent or even absent in these more complex patterns at small $k$. These observations highlight the role of structural heterogeneity and facet geometry in shaping the rigidity percolation process in general origami systems.


Nevertheless, the power of the Most Efficient selection rule, along with larger values of $k$-candidate pool, helps avoid the selection of triangular facets in the early stages, thereby restoring the linear regime even in individual simulations (see the red regions in Fig.~\ref{fig:rule1} where a linear decrease is observed across general structures and where the linearity becomes more pronounced as $k$ increases). The candidate sampling mechanism under high  $k$ offers a larger pool of candidate facets, thereby increasing the probability of selecting facets with more vertices at the early stage. Imposing planarity on these facets is more effective in reducing the degrees of freedom, as it constrains more vertex motions and indirectly influences adjacent facets. As a result, at larger $k$, a linear early-stage decrease in DOF followed by a flatter region similar to the Miura-ori is observed.


Since imposing a planarity constraint on a facet with more vertices tends to reduce the DOF more significantly, the facet type plays a crucial role in the evolution of DOF. In periodic origami structures, due to their inherent periodicity and foldability, the patterns are typically composed of quadrilaterals and triangles. This simple facet composition leads to two distinct behaviors: selecting a quadrilateral facet in the early stage decreases the DOF by one (as it contributes a single constraint row to the rigidity matrix $A$, as discussed in the previous section), while selecting a triangular facet does not reduce the DOF. As a result, under the Most Efficient selection rule with large $k$, the system tends to prioritize quadrilateral facets, leading to a linear decrease in DOF until the final value is reached, after which the curve flattens (see the first-row plots of Fig.~\ref{fig:rule1}).


In rotational origami patterns, central or sub-central facets with many edges and vertices are often present to achieve central rotational symmetry. For example, the Lang Oval pattern features an 18-gon, Hex/Tri features a hexagonal sub-center, and Lang Honeycomb features a central hexagon. To achieve rotational symmetry, these patterns incorporate a mix of triangular, quadrilateral, and other polygonal facets, resulting in more than two types of facets in the structure. Therefore, under the Most Efficient selection rule, combined with the power of choices, the selection process tends to prioritize hexagonal facets over quadrilaterals and quadrilaterals over triangles. As a result, the DOF evolution of rotational origami with a larger $k$ often exhibits a piecewise linear trend, with an initial steep linear decay, followed by a slower linear regime, and finally transitioning into a nonlinear region as the remaining unconstrained facets (mainly triangular facets) will hardly affect the DOF decrease (see the second-row plots of Fig.~\ref{fig:rule1}).


In perforated origami structures, due to the presence of cuts between facets, the individual facet planarity is less related to each other. As a result, imposing a planarity constraint on one facet does not significantly affect the planarity of others, compared to connected origami structures. Hence, the DOF decays linearly as non-triangular facets are increasingly prioritized with larger values of $k$. Once only triangular facets remain, the DOF curve flattens as it reaches the final DOF value (see the third-row plots of Fig.~\ref{fig:rule1}). Notably, in the Kirigami Honeycomb structure, all simulation results for different values of $k$ overlap as a linear line in the plot. This is because the structure is highly floppy and contains no linearity-breaking triangular facets, and enforcing planarity on one facet does not implicitly affect the planarity of neighboring facets at any level, due to the frequent presence of cuts between them. As a result, the DOF decreases linearly throughout the process until the end.


To obtain a more comparative conclusion, we observe that, in general, the linear regime of the DOF with a relatively large $k$ tends to decrease as the ratio of triangular facets increases across different patterns.
Additionally, the standard deviation across simulations decreases as $k$ increases, indicating that the power of choices has a strong effect, leading to a narrower transition width in rigidity percolation.
This effect is particularly significant for patterns with a higher proportion of triangular facets, where the selection rule becomes more influential in avoiding the early selection of facets that do not reduce the structure’s overall DOF.

\begin{figure}[t!]
    \centering
    \includegraphics[width=\linewidth]{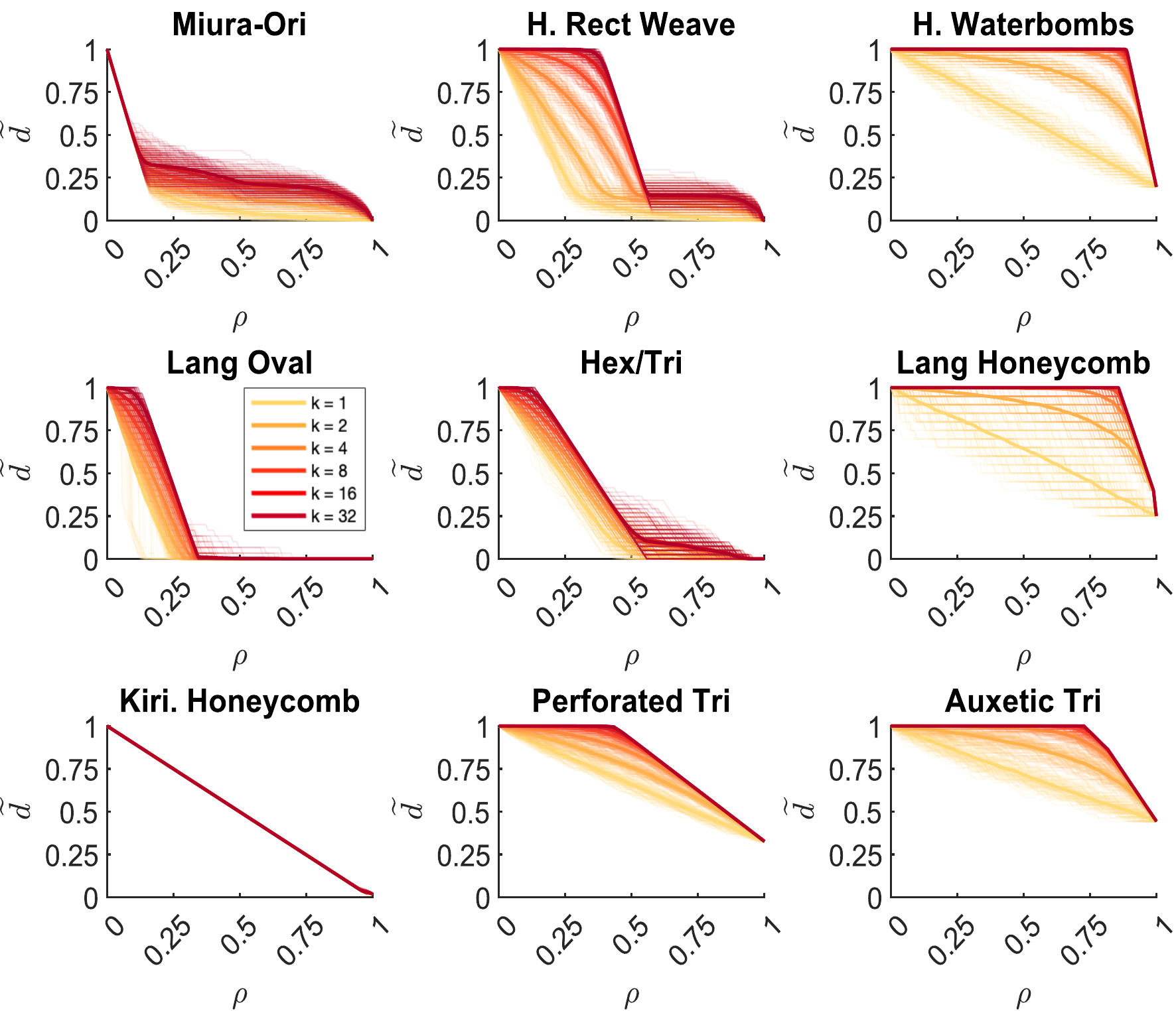}
    \caption{\textbf{Change in the normalized DOF under the Least Efficient selection rule with different numbers of choices for nine types of origami structures.} The three rows correspond to the Periodic Origami, Rotational Origami, and Perforated Origami structures, respectively. Within each row, the ratio of triangular facets increases from left to right across the structures.}
    \label{fig:rule2}
\end{figure}

\subsection{The Least Efficient selection rule}
Similar to the Most Efficient selection rule, Fig.~\ref{fig:rule2} shows $\widetilde{d}$ from 100 simulations under the Least Efficient selection rule across different $k$ values for all nine types of origami patterns. The light lines indicate the individual 100 simulations. The dark lines indicate the average of the 100 simulations (see Appendix~\ref{appendix:DOF} for more plots with different pattern sizes).


It is easy to see that the DOF evolution of general origami structures at the early stage behaves quite differently compared to Miura-ori under the Least Efficient selection rule. Due to the presence of triangular facets, in individual simulations, the rule often selects triangular or small-vertex facets at the early stage, introducing nonlinearity in the DOF decay in the individual simulations. With a larger $k$, the increased sampling pool gives a better chance to select triangular facets in the early stage, which further slows down the DOF reduction. Unlike Miura-ori and Kirigami Honeycomb, which consist only of quadrilateral facets and exhibit inevitable linear DOF decay at early stages under any rule, structures with triangular facets can delay DOF reduction in the early stage. Under the Least Efficient selection rule with large $k$, a slower or even flat DOF transition in the early stage can be observed (see the red regions of the corresponding structures in Fig.~\ref{fig:rule2}, where the DOF remains flat during the early stage).

\begin{figure*}[t!]
    \centering  
    \includegraphics[width=0.9\linewidth]{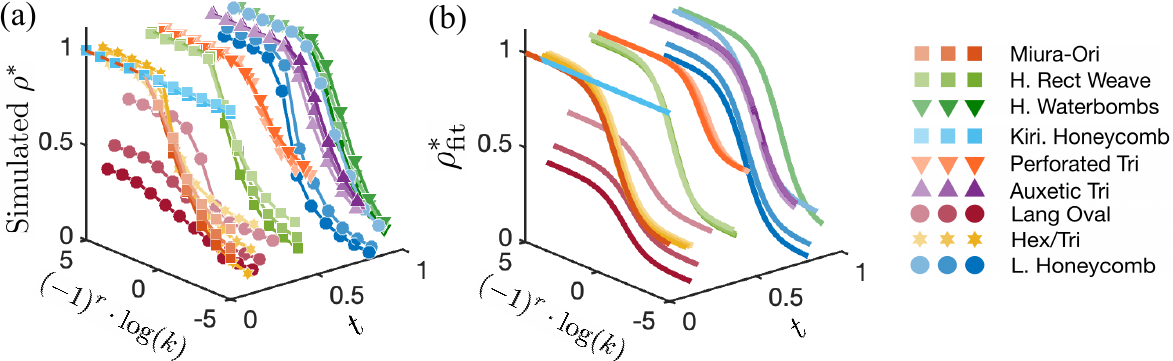}
    \caption{\textbf{The critical transition density for the nine different types of origami structures under different selection rules and different number of choices $k$.} (a) The critical transition density $\rho^{*}$ obtained from our simulations. The nine different types of origami structures are represented using different colors and marker styles. For each type of origami structure, we consider three different resolutions, represented by different marker transparencies. (b) The fitted $\rho^*$ obtained using the proposed model in Eq.~\eqref{eqt:fitting}.}
    \label{fig:critical_rho_3d}
\end{figure*}


In structures where the final DOF is one, such as Huffman Rectangular Weave and Hex/Tri, the DOF trajectory then resembles that of Miura-ori, showing an initial linear decay followed by a nonlinear regime. This can be seen from the DOF evolution of the corresponding structures in Fig.~\ref{fig:rule2}, where the DOF remains flat at the early stage, then decreases linearly, and is eventually followed by a nonlinear region. The middle decreasing and final nonlinear regions resemble the DOF trend observed in Miura-Ori, starting from its initial linear regime. Since the triangular facets have already been selected, the algorithm inevitably begins selecting non-triangular facets, leading to further DOF reduction. Once a sufficient number of polygonal facets are selected, the Least Efficient selection rule begins to preferentially select ``redundant" facets for which the planarity constraints no longer impact the DOF.


For patterns that retain multiple degrees of freedom at the final stage, the structure often contains a large proportion of triangular facets or cuts, or is composed of groups of polygonal facets that include a mix of triangular, quadrilateral, and hexagonal facets. Since the constraint selection process under the Least Efficient selection rule initially targets facets whose planarity constraints have minimal effect on the global DOF, those with fewer vertices are thus more likely to be selected in the early stage, resulting in an initial flat region. As the process progresses, it inevitably shifts toward larger or more connected facets, which constrain more degrees of freedom and lead to a more rapid reduction in DOF. As a result, in the Auxetic Triangle, Perforated Triangle, and Lang Honeycomb patterns, the DOF tends to decrease directly to the final DOF after initial flatness. This can be observed in the DOF evolution of the corresponding structures in Fig.~\ref{fig:rule2}, where the DOF remains flat at the early stage, then decreases linearly to the final DOF, or is followed by a steeper linear decrease leading to the final DOF. By contrast, the Kirigami Honeycomb structure is highly floppy. As explained earlier, under both rules, enforcing planarity on one facet does not implicitly constrain other facets due to the presence of cuts, resulting in no redundant constraints. Under the Least Efficient selection rule, the DOF still decreases linearly throughout the process until the end.


In conclusion, for patterns with triangular facets and relatively large $k$, the DOF curve typically begins with a flat regime, which can not be observed in Miura-ori due to its purely quadrilateral configuration. Then DOF evolution is followed by a (piecewise) linear regime, where the number of linear segments depends on the variety of polygon types in the origami structure. A redundant nonlinear regime may or may not appear, depending on the specific geometry and connectivity of the pattern, and whether it leads to a single-DOF structure when all planarity constraints are imposed.

\subsection{A unified model for the rigidity percolation
transition}

To study the rigidity percolation transition in different origami structures, we consider the probability of getting a minimum-DOF structure (i.e., a structure achieving the minimum possible DOF $d_{\text{final}}$) for each planarity constraint density $\rho$, defined as
\begin{equation}
    P(\rho) = \frac{\text{Number of minimum-DOF structures at $\rho$}}{\text{Total number of simulations}}.
\end{equation}
To quantify how the change in the number of choices $k$ affects the rigidity percolation transition in different patterns and sizes,  we define the \emph{critical transition density} $\rho^*$ as the minimum $\rho$ with the probability of getting a minimum-DOF structure $P \geq 1/2$ in our simulations (see Appendix~\ref{appendix:DOF}
for plots of $P$ vs $\rho$ for different origami structures). For each origami structure, since the triangular facet ratio is a key variable in the rigidity percolation study, we record the triangular facet ratio and its corresponding critical transition density under different selection rules and numbers of choices $k$. Detailed numerical data can be found in Appendix~\ref{appendix:rho}.


In Fig.~\ref{fig:critical_rho_3d}(a), we present a 3D plot of the critical transition density $\rho^*$ against the selection parameter $(-1)^r \cdot \log(k)$, where $r = 1$ corresponds to the Most Efficient selection rule and $r = 2$ to the Least Efficient selection rule. We also include the triangular facet ratio $t$ as a structural parameter in the 3D plot. Note that when $k = 1$, both rules reduce to fully stochastic selection since only one candidate is available. In this case, the selection parameter $(-1)^r \cdot \log(k)$ equals zero, and the critical transition density $\rho^*$ corresponds to the average of the $\rho^*$ values obtained under the two rules. See also Appendix~\ref{appendix:rho} and Supplementary Video 3 for additional results and visualizations.


From the 3D plot, we observe that for rotational origami structures, smaller resolutions tend to exhibit higher critical transition densities, as increasing the resolution changes the component ratio of the polygonal facet types. In contrast, resolution has little effect on periodic and perforated origami structures, since repeating the unit cell does not change the composition and the component ratio of the polygonal facets type. We conclude that if the change of pattern resolution does not significantly affect the facets ratio, then it does not significantly affect the critical transition density.


It is noteworthy that the underlying philosophies of the two rules are fundamentally reversed as $k$ increases on both sides. At $k = 1$, both rules behave as fully stochastic selection strategies. As $k$ increases from this baseline, the Most Efficient selection rule tends to favor selecting multi-vertex facets, while the Least Efficient selection rule increasingly avoids them. To highlight this contrast, data points from both rules are combined into a single plot, as shown in Fig.~\ref{fig:critical_rho_3d}(a).


Moreover, note that the more a facet is avoided under one rule, the more likely it is to be selected under the other. Avoiding certain types of facets is intuitively easier: if a facet does not appear in the $k$-candidate pool (especially when $k$ is small), it is automatically avoided. Even if it appears, the rule can be used to steer selection away from it. Therefore, under the Least Efficient selection rule, both the rule itself and smaller values of $k$ contribute to avoiding undesired facets. As a result, even small values of $k$ can significantly delay the reduction of degrees of freedom (DOF), producing a similar effect as larger $k$.

In contrast, actively favoring a specific facet type is more difficult than simply avoiding it. It first requires the facet to appear in the $k$-candidate pool through stochastic sampling, and then be selected. Thus, to ensure consistent selection of desired facets, $k$ must be sufficiently large to provide enough candidate options. Therefore, while $k$ influences both rules, its effect is more sustained under the Most Efficient selection rule.

The duality and above contrast of two rules explain the rapid change in the critical transition density $\rho^*$ within $k \in [-4, 8]$, with a symmetry center slightly biased toward the Most Efficient selection rule. Beyond this range, $\rho^*$ remains relatively stable across all structures.

Recognizing the symmetric behavior in the simulation results, we observe that $\rho^*$ varies significantly around a central value of $k$, with the center of symmetry slightly shifted to the right. More specifically, all curves appear to be centered around the right of the point $(-1)^r \cdot \log(k) = 0$, i.e., when $k = 1$, and origami structures with different triangular facet ratios $t$ give different center values. Moreover, in general, the center value shows clearly different trends for different $t$. This suggests that $\rho^*$ partially depends on a function of $t$. Additionally, for both the Most Efficient selection rule ($r = 1$) and the Least Efficient selection rule ($r = 2$), the simulated $\rho^*$ values tend to stabilize as $k$ increases. Hence, we fit the simulation results using a $\tanh$-based model with parameters that control both the steepness and the center of symmetry. Specifically, we consider the following model:
\begin{equation} \label{eqt:fitting}
\rho_{\text{fit}}^*(r,k,t) = a \cdot \tanh \left(b \cdot (-1)^r \cdot \log(k)+c \right) + d t + f,
\end{equation}
where $k$ is the number of candidate facets, $r \in \{1, 2\}$ denotes the rule type, $t$ is the triangular facet ratio, and $a$, $b$, $c$, $d$, $f$ are fitting parameters. By fitting this model to each origami structure and resolution, we see that the fitted result $\rho_{\text{fit}}^*$ matches the simulated values $\rho^*$ both qualitatively and quantitatively (see Fig.~\ref{fig:critical_rho_3d}(b) and the detailed results in Appendix~\ref{appendix:rho}). Thus, for any given origami pattern and resolution, we can predict the critical transition density using the corresponding fitted parameters.

The critical transition in rigidity percolation marks the point at which an origami structure becomes mechanically rigid or significantly less flexible, and serves as a key design target for applications such as mechanical memory, tunable stiffness, and reconfigurable metamaterials. For high-resolution structures, direct simulation is computationally expensive. Instead, by computing the triangular facet ratio $t$ of high-resolution origami structures, one can use the fitted model to predict the critical transition density $\rho_{\text{fit}}^*(r, k, t)$ for given values of $r$ and $k$. One can also use the fitted model in a reversed way to achieve a desired critical transition density. Given a high-resolution origami structure with triangular facet ratio $t$, the model allows determining suitable values of the selection rule $r$ and the number of choices $k$ required to reach the target rigidity. In physical rigidity-based origami, if a desired number of rigidified facets is specified as $N_r = \rho_{\text{fit}}^*(r, k, t) \times \text{(total number of facets)}$, one can select appropriate values of the selection rule $r$ and the number of choices $k$ accordingly. By referring to the recorded selection history from rigidity percolation simulations, these constraints can then be directly applied to the physical structure. This approach enables the physical construction of a significantly less flexible origami design that contains the desired number of rigidified facets.

\section{Discussion}
In this work, we have studied the rigidity control of various types of origami structures with different periodicity properties, rotational symmetries, and topologies. In particular, we have considered how different selection rules in changing the facet planarity will affect the rigidity and percolation transitions of different origami structures. We have shown that the changes in the degrees of freedom of periodic origami structures generally exhibit a combination of linear and nonlinear regimes similar to that of the well-known Miura-Ori structures. For rotational origami structures, the central facet plays an important role and can hence lead to a larger variation in the rigidity transition. By contrast, for perforated origami structures, the individual facet planarity conditions are less related to each other due to the cuts, and hence the DOF will generally decrease straightly under the Most Efficient selection rule, until only the ``redundant'' facets are left, and the trend under the Least Efficient selection rule will be the opposite.

More generally, one can see that in all origami structures, the rigidity control is highly relevant to the presence of triangular facets. Moreover, by analyzing the critical transition density $\rho^*$, one can observe the duality between the Most Efficient selection rule and the Least Efficient selection rule. The relationship between the number of choices $k$, the type of the selection rule $r$, the triangular facet ratio $t$, and the critical transition density $\rho^*$ can be described by a simple model involving a hyperbolic tangent function and some other linear terms. Altogether, this paves a new way for the analysis of the rigidity control of general origami structures and other art-inspired mechanical metamaterials.

From a statistical viewpoint, starting with the hypergeometric model, we estimate the probability of DOF change at a given $\rho$ and show that the triangular facet ratio is a key variable. A natural next step is to define additional concepts and introduce necessary assumptions into the stochastic process of random sampling of $k$-candidates and selection. One can then study the expectation $\mathbb{E}[\rho^*]$, variance $\text{Var}[\rho^*]$, where $\rho^* = \arg\min_\rho d(\rho) = d_{\text{final}}$ and further model the simulation process using a stochastic differential equation framework. This provides more theoretical insights into the critical transition and its width. Another natural next step is to extend our study to the rigidity control of curved fold origami structures~\cite{dias2012shape,liu2024design,chai2024programmable,song2024review} and other two- and three-dimensional structural assemblies~\cite{lubbers2019excess,overvelde2017rational,choi2020control}.

\bibliographystyle{ieeetr}
\bibliography{reference}

\clearpage

\centerline{\large\textbf{Supplementary Information}}
\appendix
\renewcommand\thefigure{S\arabic{figure}}    
\setcounter{figure}{0}
\renewcommand\thetable{S\arabic{table}}    
\setcounter{table}{0}


\section{The geometry of different origami structures} \label{appendix:origami}
In the main text, we provided one representative view of all types of origami structures studied in this work. Fig.~\ref{fig:SI_manyori} shows additional views of each type of origami structure to help readers better understand their geometry and structure.

Also, it is noteworthy that the calculation of the infinitesimal rigidity matrix $A$ depends on the vertex coordinates of the origami structure. As discussed in~\cite{li2025explosive,chen2019rigidity}, changes in geometric parameters (i.e., the folding percentage) of the Miura-ori do not affect the rigidity percolation behavior. It is natural to ask whether the geometric parameters of other origami structures influence their rigidity percolation behavior. 

To address this, we present in Fig.~\ref{fig:SI_geometry_recweave}(a) the simulation results for the Huffman Rectangular Weave origami structure with 129 facets, a folding percentage of 25\%, and $k = 1, 2, 4, 8, 16, 32$ under both selection rules. We then increase the folding percentage to 50\% and 75\% while keeping the configuration and resolution unchanged, and the corresponding simulation results (100 simulations for each $k$) are shown in Fig.~\ref{fig:SI_geometry_recweave}(b)–(c). As shown on the left of Fig.~\ref{fig:SI_geometry_recweave}, the geometries with 50\% and 75\% folding percentages differ significantly from the 25\% case considered in the main text. Nevertheless, the 3$\times$4 plots on the right demonstrate that, under both the Most Efficient and Least Efficient selection rules, the simulation results for the evolution of the normalized DOF and rigidity percolation remain highly consistent. In particular, the trend of the normalized DOF evolution and the increasing sharpness in the transition of the probability $P$ of obtaining a minimum-DOF structure are similar across different folding percentages and rules as $k$ increases. A comparable trend is also observed in Fig.~\ref{fig:SI_geometry_waterbomb} for the Huffman Waterbombs structure with different folding percentages (25\%, 50\%, and 75\%), confirming the robustness of the rigidity percolation behavior across different geometric configurations. Therefore, we conclude that the DOF evolution and explosive rigidity percolation transition are independent of the geometry of the general origami structure.

Besides, in the main text, we stated a theorem regarding the number of sub-planarity constraints required to control the planarity of an $n$-sided polygonal facet. Below, we give the detailed proof of the theorem.

\begin{theorem}
The number of sub-planarity constraints required to control the planarity of an $n$-sided polygonal facet is exactly $n - 3$. 
\end{theorem}

\begin{proof}
We first show that the edge and no-shear constraints contribute a rank of $n + (n - 3) = 2n - 3$ to the rigidity matrix. 

Note that the first triangle $(\mathbf{v}_1, \mathbf{v}_2, \mathbf{v}_3)$ in the triangulation of the $n$-sided polygonal facet contributes three edge constraints. These constraints are linearly independent and form a $3 \times 3n$ submatrix with nonzero entries only at columns corresponding to $\mathbf{v}_1, \mathbf{v}_2, \mathbf{v}_3$. Each subsequent triangle $(\mathbf{v}_k, \mathbf{v}_i, \mathbf{v}_j)$ introduces two new edge constraints, which again contribute two independent rows to the rigidity matrix due to their sparsity pattern (i.e., non-overlapping support in the matrix rows). These rows involve only the coordinates of $\mathbf{v}_k$, $\mathbf{v}_i$, and $\mathbf{v}_j$. There are $(n - 3)$ such additional triangles beyond the first one, contributing $2(n - 3)$ additional rows. Therefore, the edge and no-shear constraints together contribute a total rank of 
\begin{equation}
3 + 2(n - 3) = 2n - 3.
\end{equation}

Since the polygon has $n$ vertices in $\mathbb{R}^3$, it has in total $3n$ degrees of freedom. Subtracting 6 for the rigid body motions (3 translations and 3 rotations), the maximal rank attainable is $3n - 6$. Thus, the minimum number of constraints needed to reach full rank is 
\begin{equation}
3n - 6 - (2n - 3) = n - 3.
\end{equation}
Hence, we conclude that the above $(n - 3)$ planarity constraints are necessary to control the planarity of the polygonal facet. 

Moreover, note that it suffices to enforce planarity between each adjacent pair of triangles in the triangulation. Since each new triangle shares an edge with the previous one, enforcing planarity along these $n - 3$ internal edges is sufficient to ensure the entire polygon remains planar. Therefore, we conclude that the $n - 3$ planarity constraints in the above argument constitute one of the minimal sets required to control the polygonal facet. Therefore, no constraint is redundant.

\end{proof}

\section{The degrees of freedom (DOF) in the origami structures with different pattern resolutions} \label{appendix:DOF}

In the main text, we presented the rigidity percolation results for one size of each origami pattern. Here, we extend our analysis by considering multiple resolutions for each type of origami structure. 

Specifically, for the class of Periodic Origami structures (Fig.~\ref{fig:SI_sizes}(a)), in addition to the Miura-Ori with 400 facets shown in the main text, we include views of the Miura-Ori with 100 and 225 facets. For the Huffman Rectangular Weave, we include structures with 129, 313, and 577 facets. For the Huffman Waterbombs, we show structures with 178, 403, and 718 facets. For the class of Rotational Origami structures (Fig.~\ref{fig:SI_sizes}(b)), we show the Lang Oval with 69, 103, and 137 facets; the Hex/Tri with 97, 205, and 1285 facets; and the Lang Honeycomb with 91, 367, and 829 facets. For the class of Perforated Origami structures (Fig.~\ref{fig:SI_sizes}(c)), we show Kirigami Honeycomb with 72, 120, and 276 facets; Perforated Triangle with 39, 106, and 342 facets; and Auxetic Triangle with 88, 206, and 570 facets.

For all types of origami structures at different resolutions, we perform numerical simulations using the same setup as described in the main text. We compute the normalized DOF and the probability $P$ of obtaining a minimum-DOF structure at each planarity constraint density $\rho$. We further present our observations on the normalized DOF and the transition width across different resolutions and values of $k$. While some variations on the same structure across the different resolutions may exist, the overall DOF evolution trend for each structure remains consistent with the corresponding observations presented in the main text. The results provided here may serve as a reference for readers interested in specific structures and may help guide further studies on the rigidity control of particular origami designs, especially in identifying the optimal $k$ to achieve a prescribed low transition width.

We first recall the definition of the transition width. The transition width is defined as the interval of~$\rho$ over which the probability~$P$ increases from 0 to 1. More precisely, it is given by
\begin{equation}
\rho_w = \rho_1 - \rho^0,    
\end{equation}
where $\rho_1$ is the minimum $\rho$ for which $P = 1$, and $\rho^0$ is the maximum $\rho$ for which $P = 0$ in our simulations.

For the DOF evolution under the Most Efficient selection rule of the periodic origami and rotational origami structures (see Fig.~\ref{fig:SIrule1_fig1} and Fig.~\ref{fig:SIrule1_fig2}), increasing the resolution reduces the proportion of facets needed to reach the final DOF. For the DOF evolution of perforated origami structures (see Fig.~\ref{fig:SIrule1_fig3}), increasing the resolution does not affect the proportion of facets needed to reach the final DOF. In all of the above-mentioned origami structures, while the normalized DOF decreases as the resolution increases, the final DOF remains unchanged. The overall trend of DOF evolution with increasing $k$ is consistent across all resolutions. Specifically, as $k$ increases, the DOF in individual simulations tends to drop more rapidly in the early stages, eventually exhibiting a sharp and consistent linear decline across simulations until reaching the final DOF. 

For the transition width under the Most Efficient selection rule of periodic origami structures, increasing $k$ from small values (e.g., $k = 1, 2$) leads to a slight increase of transition width of the probability $P$ from 0 to 1, while further increasing $k$ results in a sharper transition width (see Fig.~\ref{fig:SI_P_rule1_fig1}). This behavior is consistent with observations in the Miura-ori structure. More explanations and quantitative analyses are provided in~\cite{li2025explosive}. Additionally, the transition width generally remains the same across different resolutions for most values of $k$. For rotational origami structures, increasing $k$ from 1 generally results in a sharper transition of the probability $P$ from 0 to 1, especially for structures with higher resolution (see Fig.~\ref{fig:SI_P_rule1_fig2}). The transition width tends to decrease with resolution for most values of $k$. For perforated origami structures, increasing $k$ from small values (e.g., $k = 1, 2$) leads to a slight increase of the transition width of the probability $P$ from 0 to 1, while further increasing $k$ results in a sharper transition width (see Fig.~\ref{fig:SI_P_rule1_fig3}). The transition width generally remains the same across different resolutions for most values of $k$.

For the DOF evolution under the Least Efficient selection rule across general origami structures (see Fig.~\ref{fig:SIrule2_fig1}, Fig.~\ref{fig:SIrule2_fig2}, and Fig.~\ref{fig:SIrule2_fig3}), the selection rule does not affect the underlying structure of the origami. Similar to the Most Efficient selection rule, while the normalized DOF decreases with increasing resolution, the final DOF remains unchanged. The overall trend of DOF evolution with increasing $k$ is consistent across all resolutions. Specifically, as $k$ increases, the DOF in individual simulations tends to remain flat for as long as possible, especially when triangular facets constitute a larger proportion of the structure, and then inevitably experience a linear decline. If the structure converges to a single-DOF configuration, a nonlinear region may appear after the linear regime due to redundancy in the facet planarity constraints left. In contrast, structures with multiple final DOFs are more likely to be reached through a consistent linear decline.

For the transition width under the Least Efficient selection rule in periodic origami structures, increasing $k$ from 1 to 2 leads to a significantly sharper transition (see Fig.~\ref{fig:SI_P_rule2_fig1}). The transition width generally remains unchanged across different resolutions for most values of $k$. For rotational origami structures, the relation between the transition width and $k$ follows a similar trend to the periodic structures. Increasing the resolution results in a slight increase in the transition width for the Lang Honeycomb structure, but a decrease for Lang Oval and Hex/Tri (see Fig.~\ref{fig:SI_P_rule2_fig2}). For perforated origami structures, increasing $k$ from 1 also sharpens the transition, although most of these structures already exhibit a sharp transition even at $k = 1$ (see Fig.~\ref{fig:SI_P_rule2_fig3}). Due to this initial sharpness and the nature of DOF decay under the Least Efficient selection rule, the transition width remains nearly 0 across different resolutions for most values of $k$.\\

\section{The critical transition density for different origami structures}  \label{appendix:rho}

In the main text, we presented a 3D plot of the critical transition density $\rho^*$ for nine types of origami structures, each with three configurations, with triangular facet ratio information provided. Here, we provide the corresponding numerical $\rho^*$ values in Table~\ref{tab:critical_density_1}, Table~\ref{tab:critical_density_2}, and Table~\ref{tab:critical_density_3}. Also, in the main text, we fit the simulated $\rho^*$ using a tanh-based model:
\begin{equation}
\rho_{\text{fit}}^* = a \cdot \tanh \left(b \cdot (-1)^r \cdot \log(k) + c \right) + d t + f.
\end{equation}
For each of the nine types of origami structures, we consider the simulation results for different parameters $(-1)^r \log(k) \in\{- \log 32, - \log 16, - \log 8, - \log 4$, $- \log 2, 0, \log 2, \log 4, \log 8, \log 16, \log 32\}$ at three different resolutions, which gives $11 \times 3 = 33$ data points. We then formulate the curve fitting problem as a constrained optimization problem and use the \texttt{fmincon} solver in MATLAB to search for the best-fit parameters $a,b,c,d,f$ with $0 \leq \rho_{\text{fit}}^* \leq 1$ based on the 33 data points. The visualization of the simulated and fitted results is shown in Fig.~\ref{fig:SI_critical_rho_3d}. To evaluate the fitting accuracy, we consider the Root Mean Squared Error (RMSE) as follows:
\begin{equation}
\mathrm{RMSE} = \sqrt{ \frac{1}{n} \sum_{i=1}^{n} \left( \rho_{\text{fit}_i}^*- \rho^*_i \right)^2 },
\end{equation}
where $n = 33$ is the total number of data points used in the fitting of one origami structure, $\rho^*_i$ is the simulated critical transition density of the $i$-th structure (with $i = 1, 2, \dots, n$), and $\rho_{\text{fit}_i}^*$ is the corresponding fitted result. For each origami structure, the fitted parameters along with the corresponding RMSE are provided in Table~\ref{tab:fit_params}. For all origami structures, the lower RMSE values indicate that the fitted models can effectively match the simulation results.

\section{Video captions}
\textbf{Supplementary Video 1}: A video showing a folding motion of the Huffman Rectangular Weave paper model.\\

\textbf{Supplementary Video 2}: A video showing an alternative folding motion of the Huffman Rectangular Weave paper model with certain facets allowed to bend.\\

\textbf{Supplementary Video 3}: An animation showing the change in the critical transition density $\rho^*$ for different origami structures under different selection rules and number of choices $k$.

\begin{figure*}[t]
    \centering
    \includegraphics[width=\linewidth]{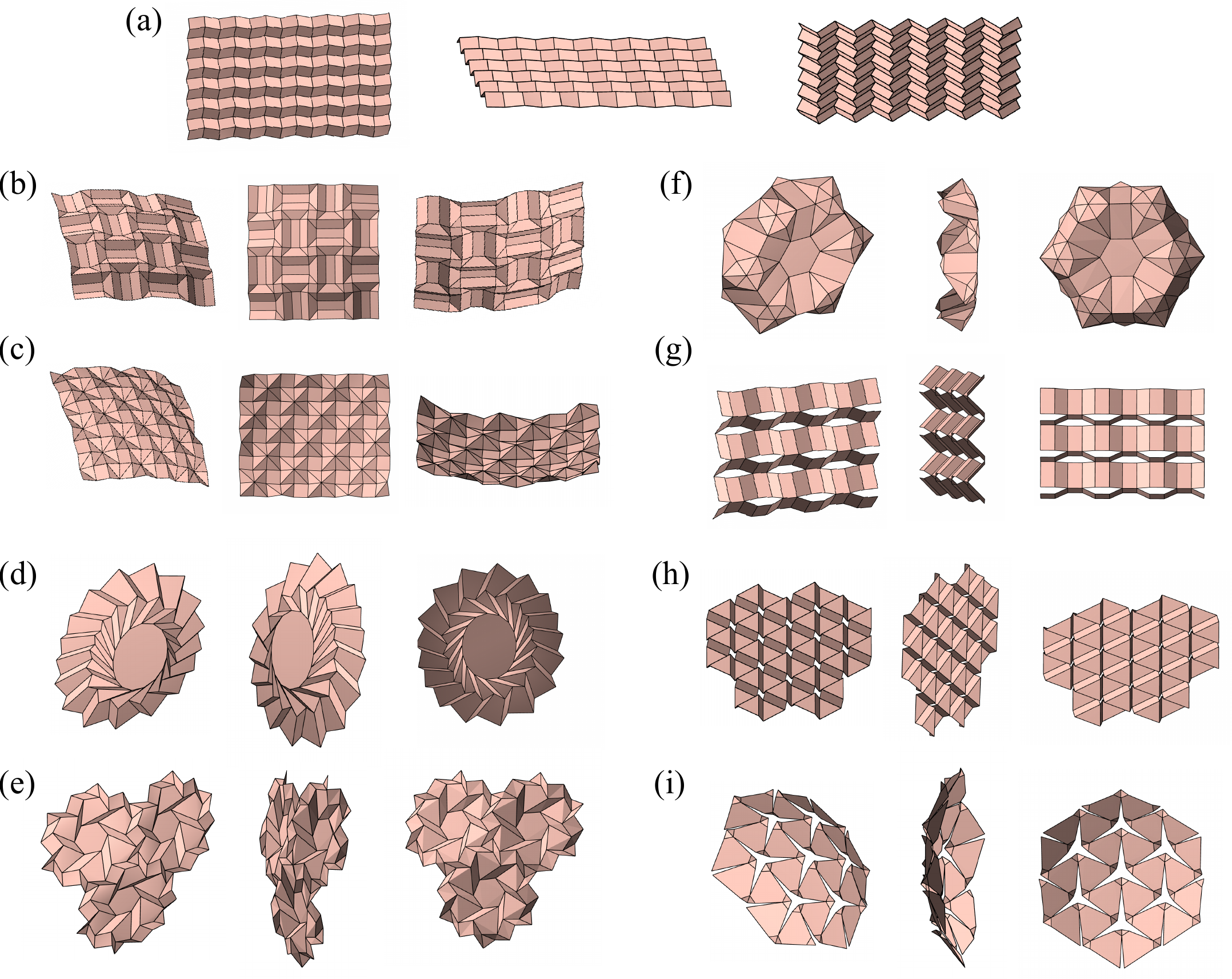}
    \caption{\textbf{Different views of the nine types of origami structures considered in our study.} For each type, three different views are provided. (a)~Miura-Ori. (b)~Huffman Rectangular Weave. (c)~Huffman Waterbombs. (d)~Lang Oval. (e)~Hex/Tri. (f)~Lang Honeycomb. (g)~Kirigami Honeycomb. (h)~Perforated Triangle. (i)~Auxetic Triangle. }
    \label{fig:SI_manyori}
\end{figure*}

\begin{figure*}[t!]
    \centering  
    \includegraphics[width=\linewidth]{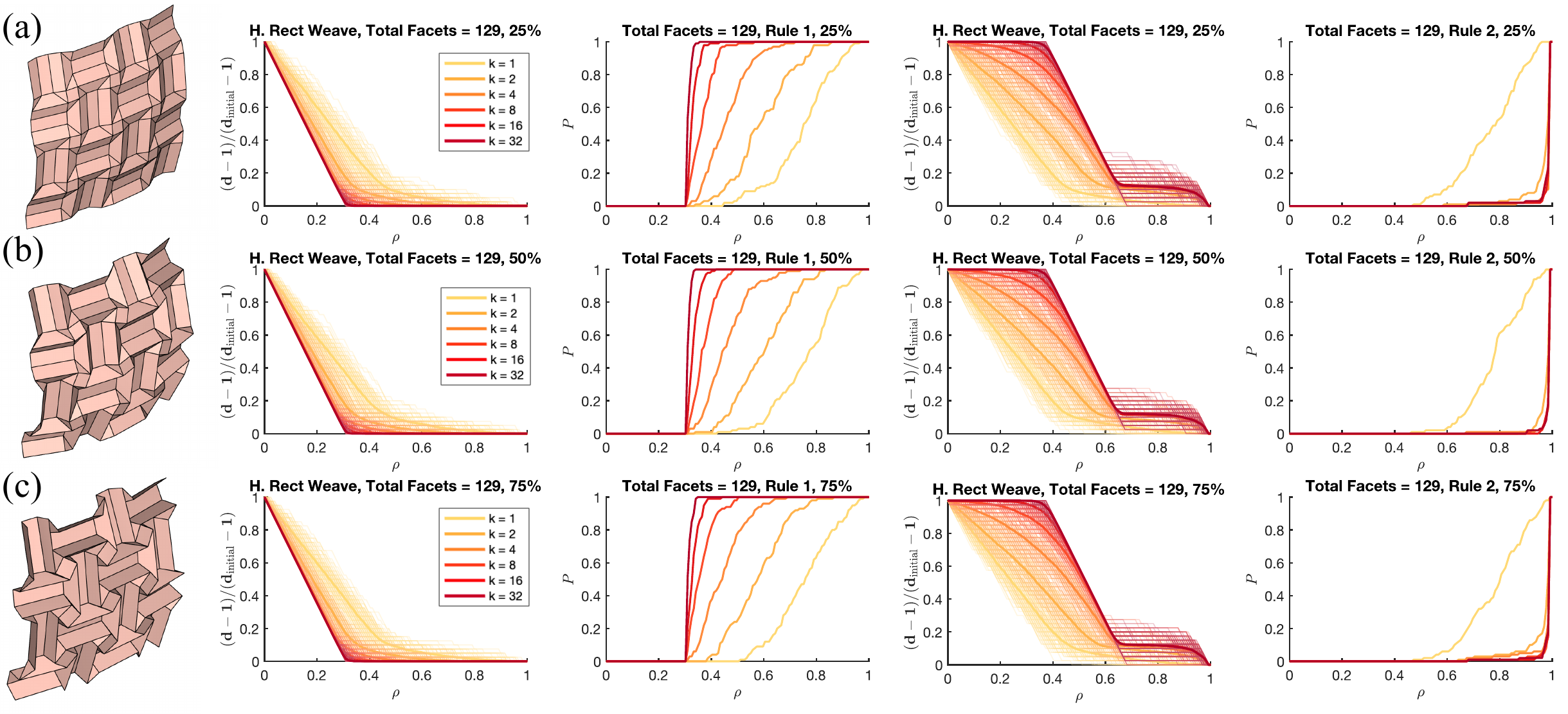}
    \caption{\textbf{Comparing the explosive rigidity percolation in the Huffman Rectangular Weave Origami structure with different folding percentage $\theta$.} (a) The results for  $\theta = 25\%$. (b) The results for $\theta = 50\%$. (c) The results $\theta = 75\%$. For each folding percentage, we consider the Huffman Rectangular Weave Origami structure (left), the rigidity percolation simulation result based on the Most Efficient selection rule with different number of choices $k$, and the simulation result based on the Least Efficient selection rule (right). Here, $\rho$ is the density of the planarity constraints explicitly imposed, and $P$ is the probability of getting a final DOF structure.}
    \label{fig:SI_geometry_recweave}
\end{figure*}

\begin{figure*}[t!]
    \centering  
    \includegraphics[width=\linewidth]{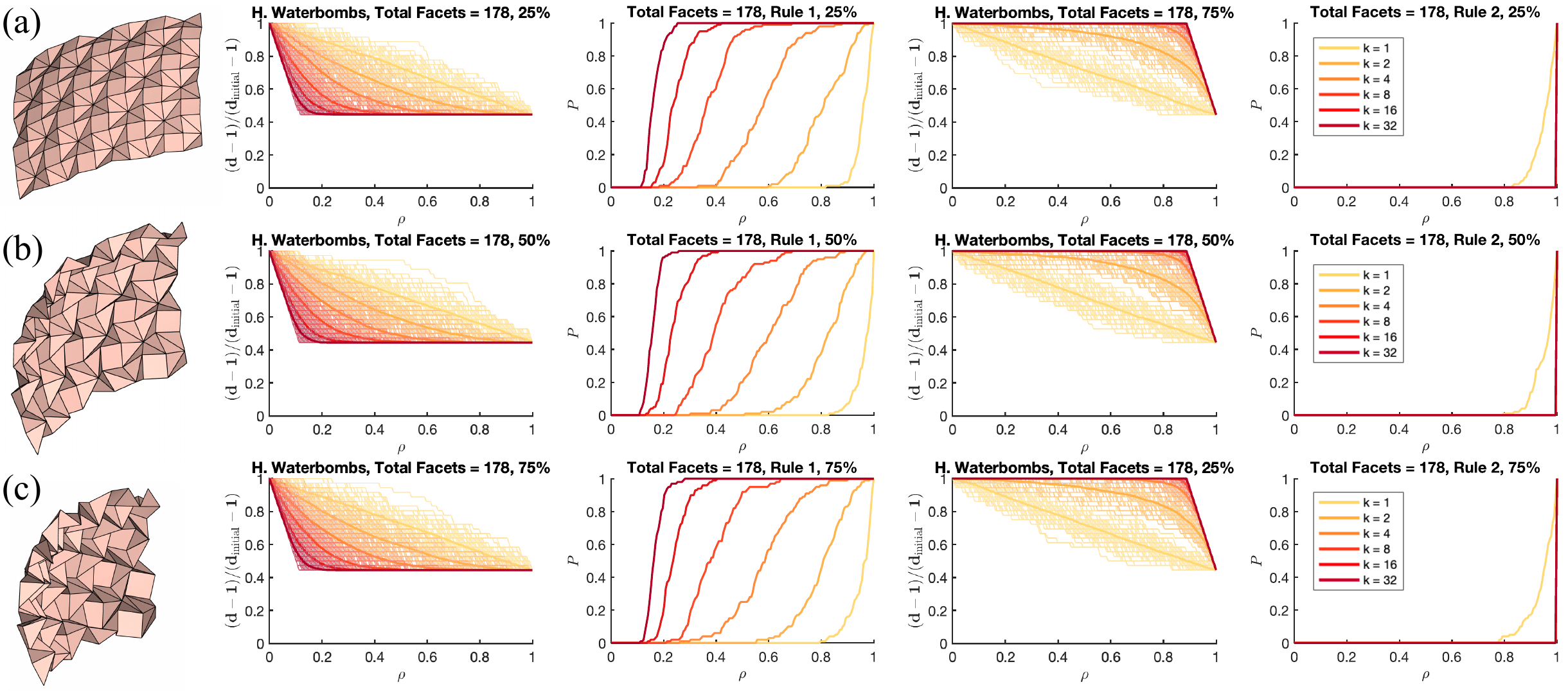}
    \caption{\textbf{Comparing the explosive rigidity percolation in the Huffman Waterbombs Origami structure with different folding percentage $\theta$.} (a) The results for  $\theta = 25\%$. (b) The results for $\theta = 50\%$. (c) The results $\theta = 75\%$. For each folding percentage, we consider the Huffman Waterbombs Origami structure (left), the rigidity percolation simulation result based on the Most Efficient selection rule with different number of choices $k$, and the simulation result based on the Least Efficient selection rule (right). Here, $\rho$ is the density of the planarity constraints explicitly imposed, and $P$ is the probability of getting a final DOF structure.}
    \label{fig:SI_geometry_waterbomb}
\end{figure*}

\begin{figure*}[t!]
    \centering  
    \includegraphics[width=\linewidth]{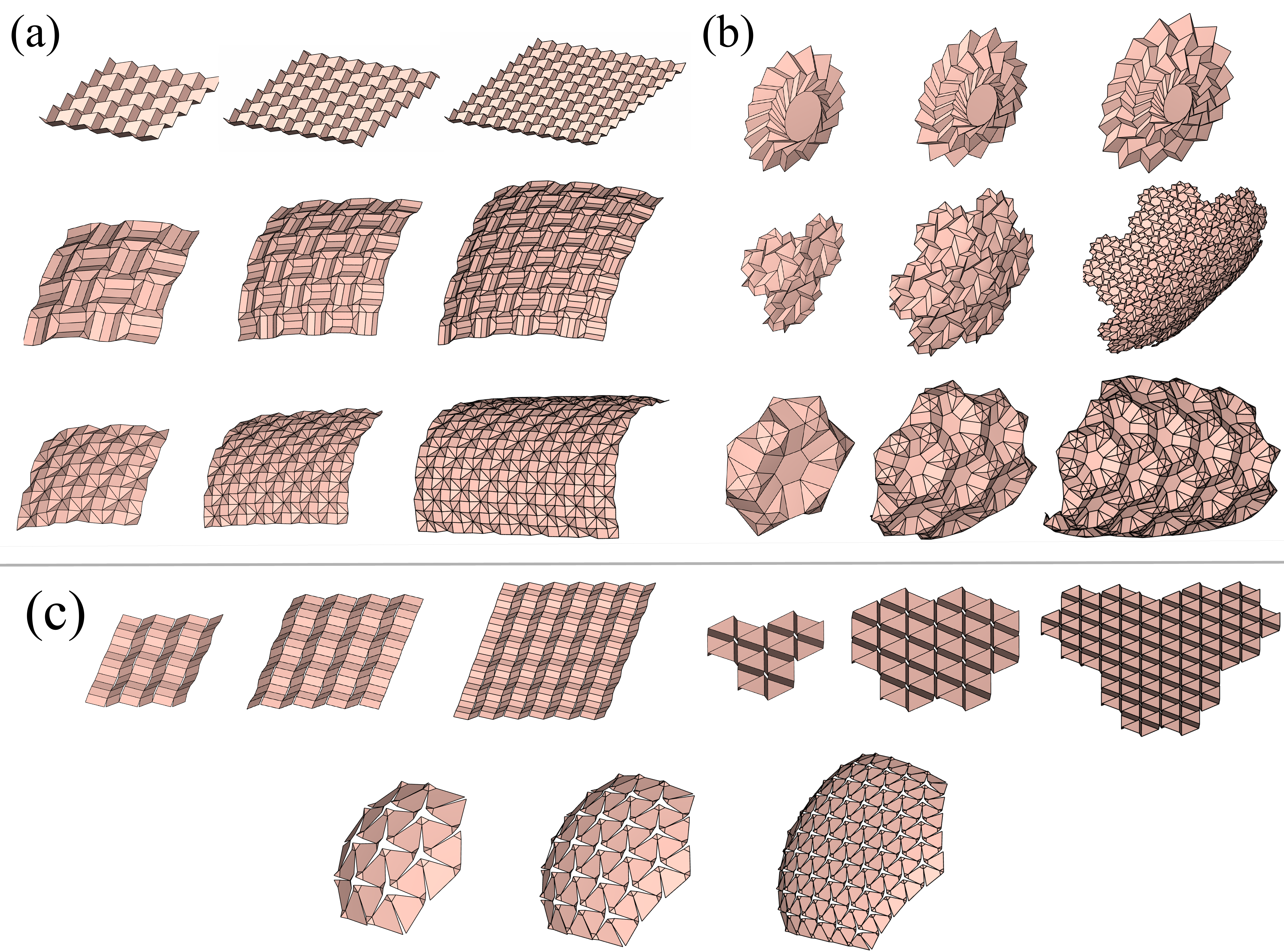}
    \caption{\textbf{Nine types of origami structures with different resolutions considered in our additional analysis.} For each type of structure, three different resolutions are considered. (a)~Periodic origami structures: Miura-ori (with 100, 225, and 400 facets), Huffman Rectangular Weave (with 129, 313, and 577 facets), and Huffman Waterbombs (with 178, 403, and 718 facets). (b)~Rotational origami structures: Lang Oval (with 69, 103, and 137 facets), Hex/Tri (with 97, 205, and 1285 facets), and Lang Honeycomb (with 91, 367, and 829 facets). (c)~Perforated origami structures: Kirigami Honeycomb (with 72, 120, and 276 facets), Perforated Triangle
    (with 39, 106, and 342 facets), and Auxetic Triangle (with 88, 206, and 570 facets).}
    \label{fig:SI_sizes}
\end{figure*}

\begin{figure*}[t!]
    \centering  
    \includegraphics[width=0.8\linewidth]{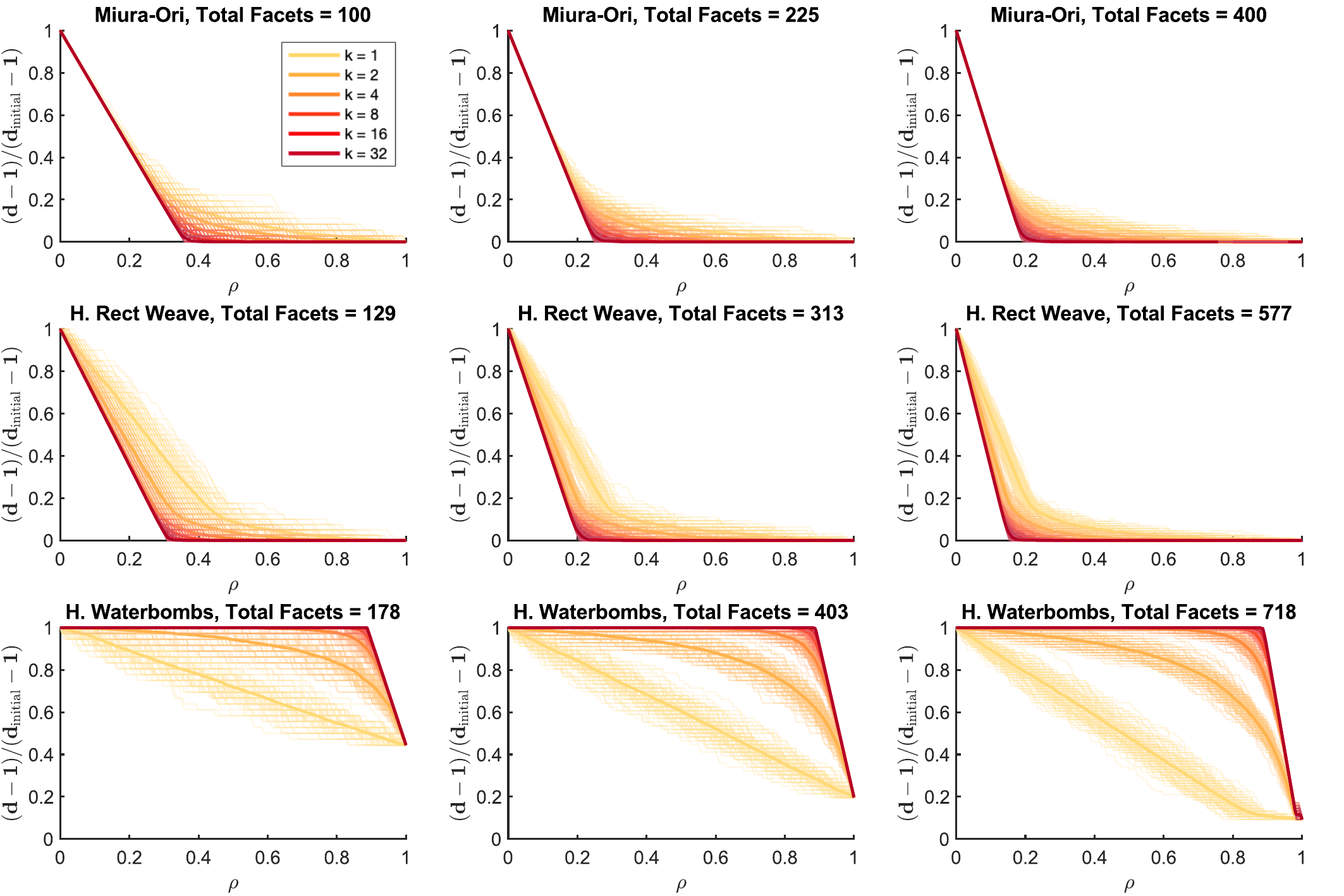}
    \caption{\textbf{Change in the normalized DOF under the Most Efficient selection rule for three types of periodic origami structures with different sizes.} For each type and each size, different numbers of choices $k = 1, 2, 4, 8, 16, 32$ are considered.}
    \label{fig:SIrule1_fig1}
\end{figure*}

\begin{figure*}[t!]
    \centering  
    \includegraphics[width=0.8\linewidth]{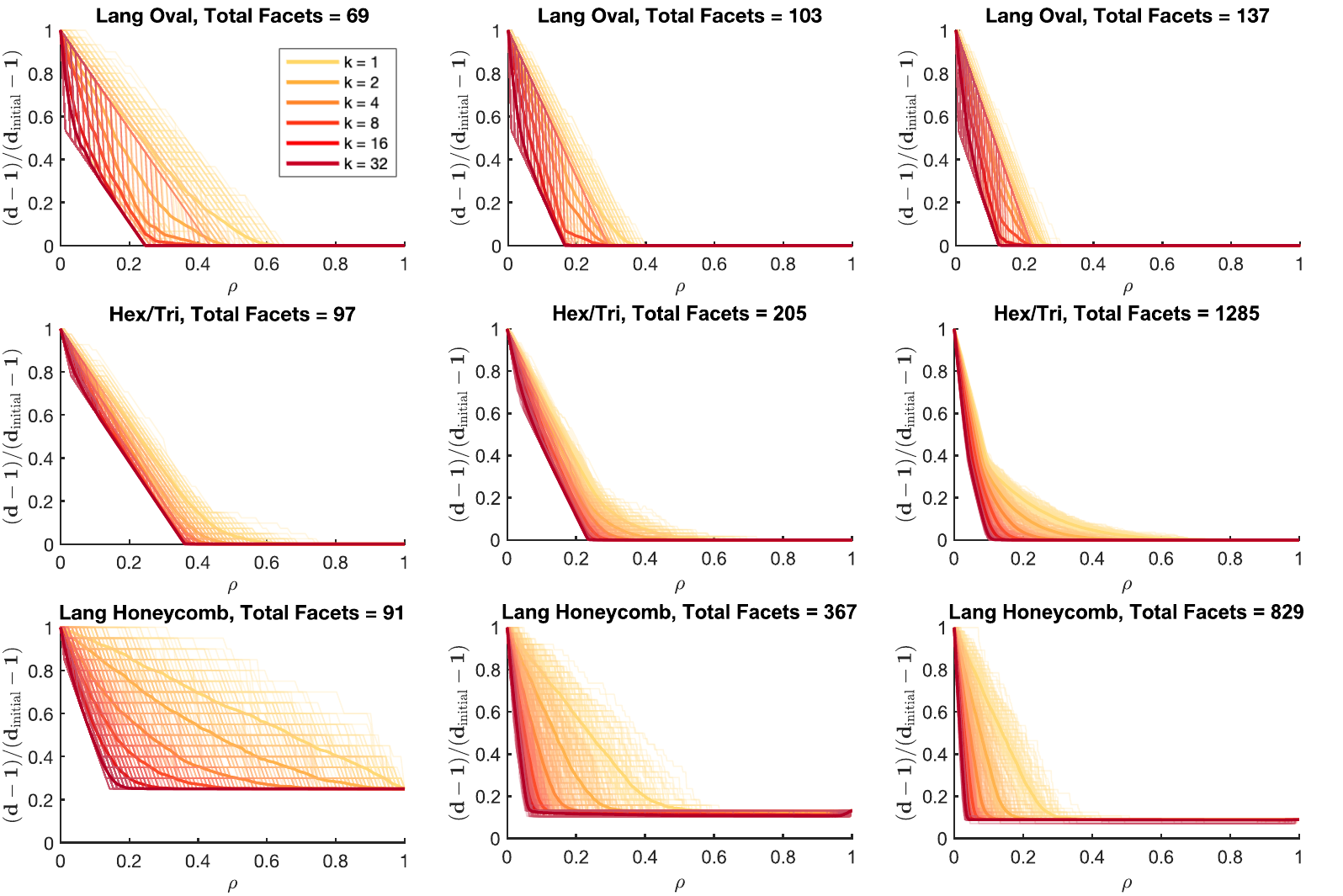}
    \caption{\textbf{Change in the normalized DOF under the Most Efficient selection rule for three types of rotational origami structures with different sizes.} For each type and each size, different numbers of choices $k = 1, 2, 4, 8, 16, 32$ are considered.}
    \label{fig:SIrule1_fig2}
\end{figure*}

\begin{figure*}[t!]
    \centering  
    \includegraphics[width=0.8\linewidth]{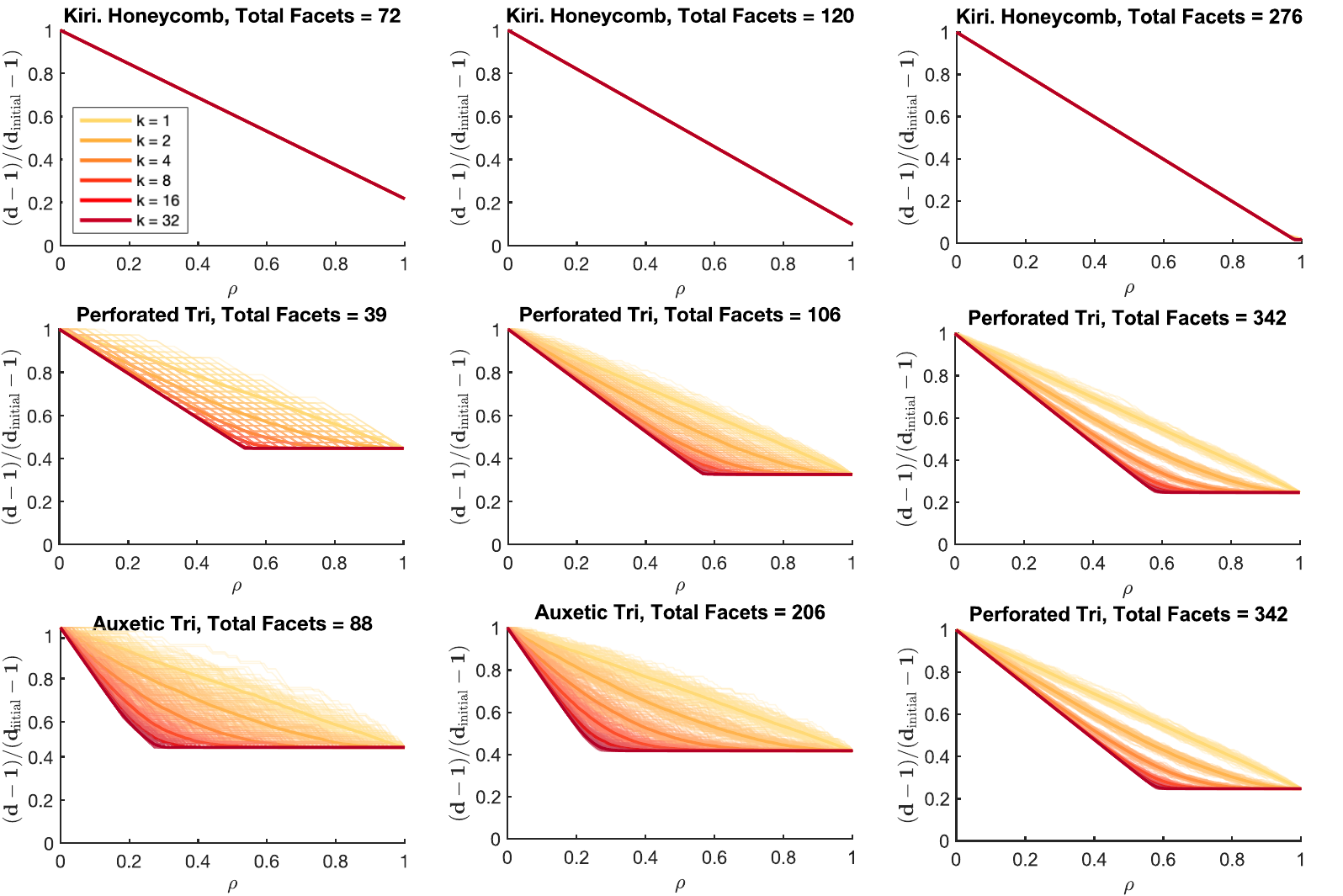}
    \caption{\textbf{Change in the normalized DOF under the Most Efficient selection rule for three types of perforated origami structures with different sizes.} For each type and each size, different numbers of choices $k = 1, 2, 4, 8, 16, 32$ are considered.}
    \label{fig:SIrule1_fig3}
\end{figure*}

\begin{figure*}[t!]
    \centering  
    \includegraphics[width=0.8\linewidth]{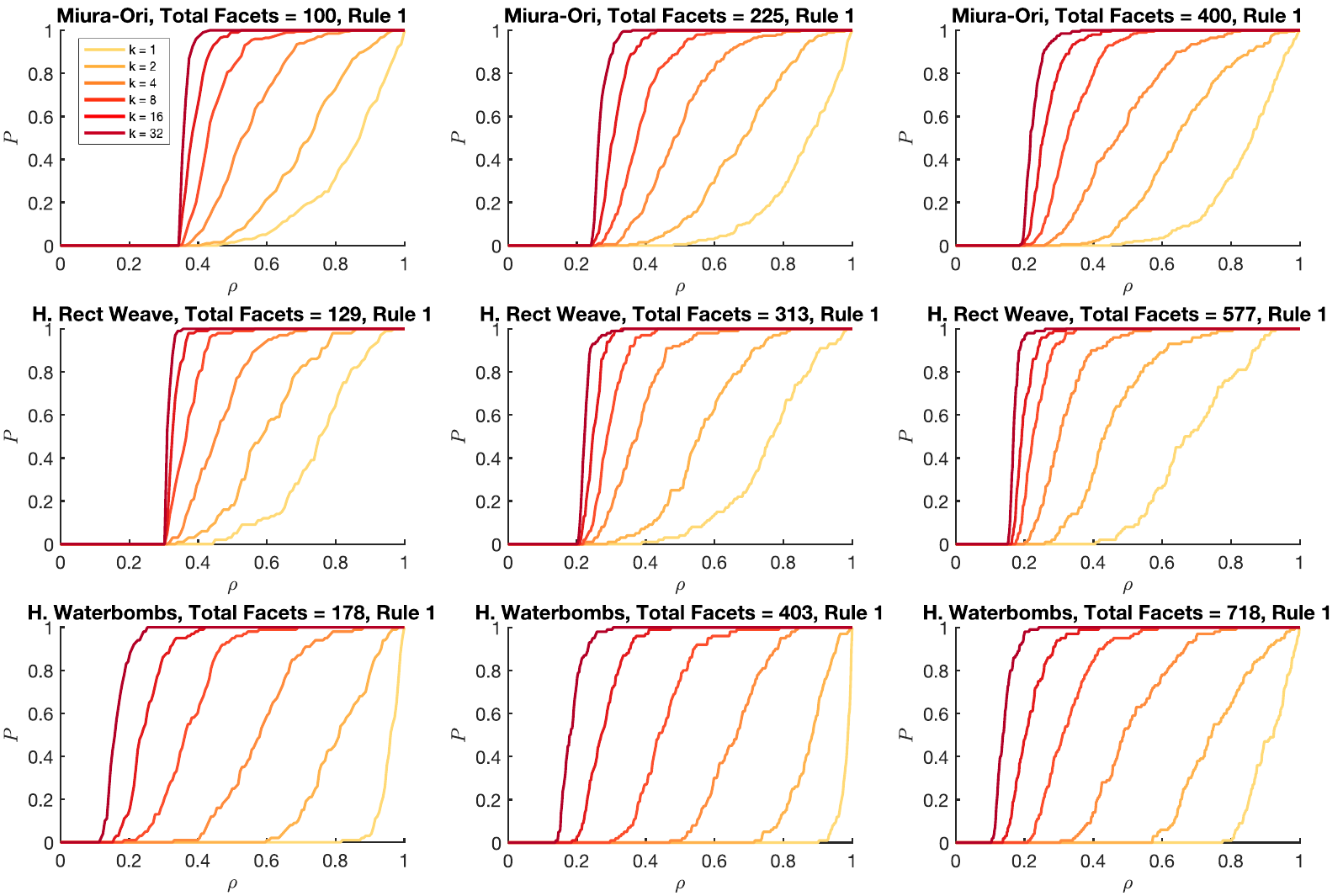}
    \caption{\textbf{Rigidity percolation in periodic origami under the Most Efficient selection rule for three origami structures with different sizes.} For different problem sizes (with the total number of facets indicated in each subfigure title) and different numbers of choices $k = 1, 2, 4, 8, 16, 32$, we compute the probability $P$ of obtaining a minimum-DOF structure at different planarity constraint densities $\rho$, based on 100 simulations.}
    \label{fig:SI_P_rule1_fig1}
\end{figure*}

\begin{figure*}[t!]
    \centering  
    \includegraphics[width=0.8\linewidth]{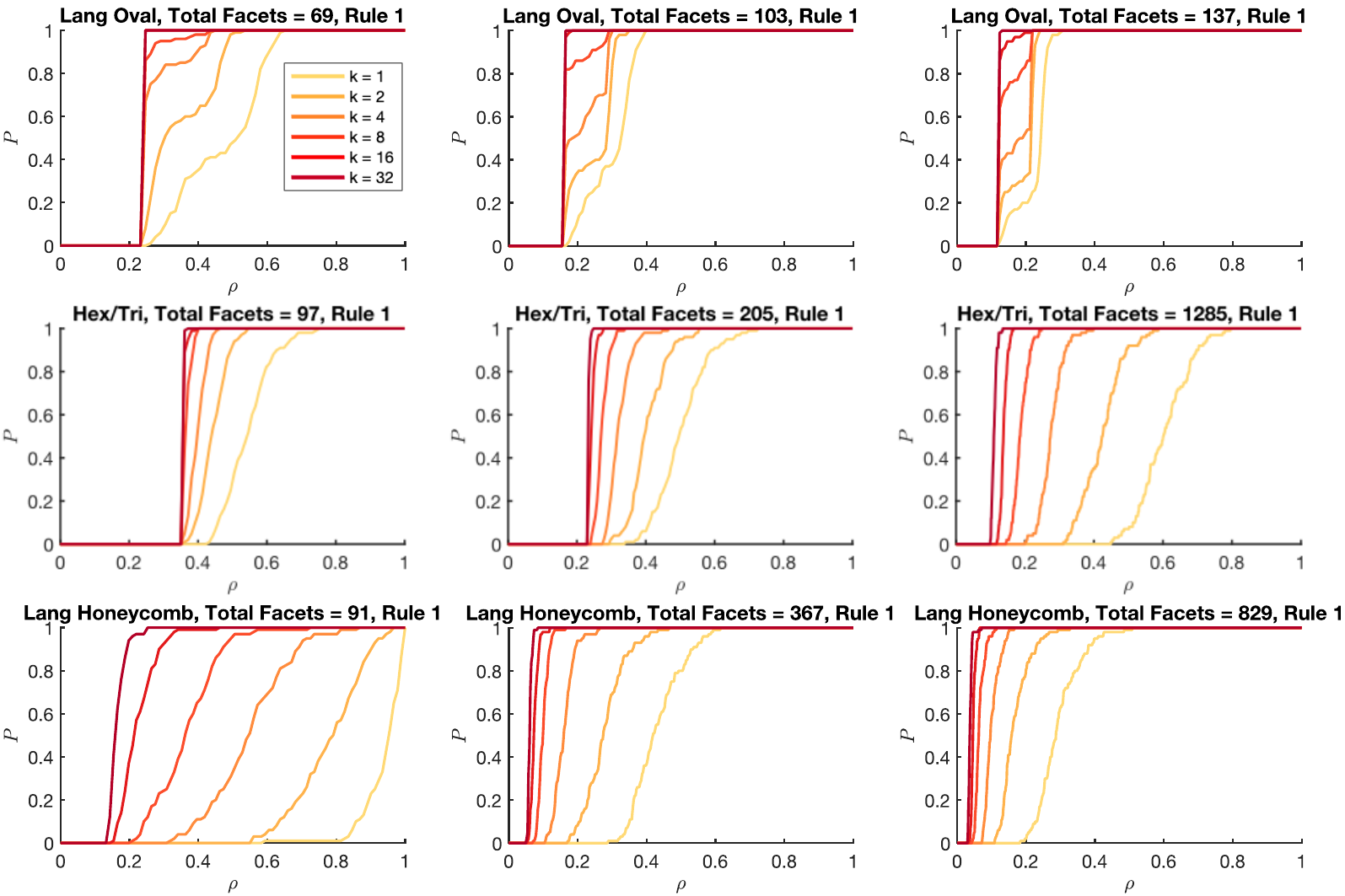}
    \caption{\textbf{Rigidity percolation in rotational origami under the Most Efficient selection rule for three origami structures with different sizes.} For different problem sizes (with the total number of facets indicated in each subfigure title) and different numbers of choices $k = 1, 2, 4, 8, 16, 32$, we compute the probability $P$ of obtaining a minimum-DOF structure at different planarity constraint densities $\rho$, based on 100 simulations.}
    \label{fig:SI_P_rule1_fig2}
\end{figure*}

\begin{figure*}[t!]
    \centering  
    \includegraphics[width=0.8\linewidth]{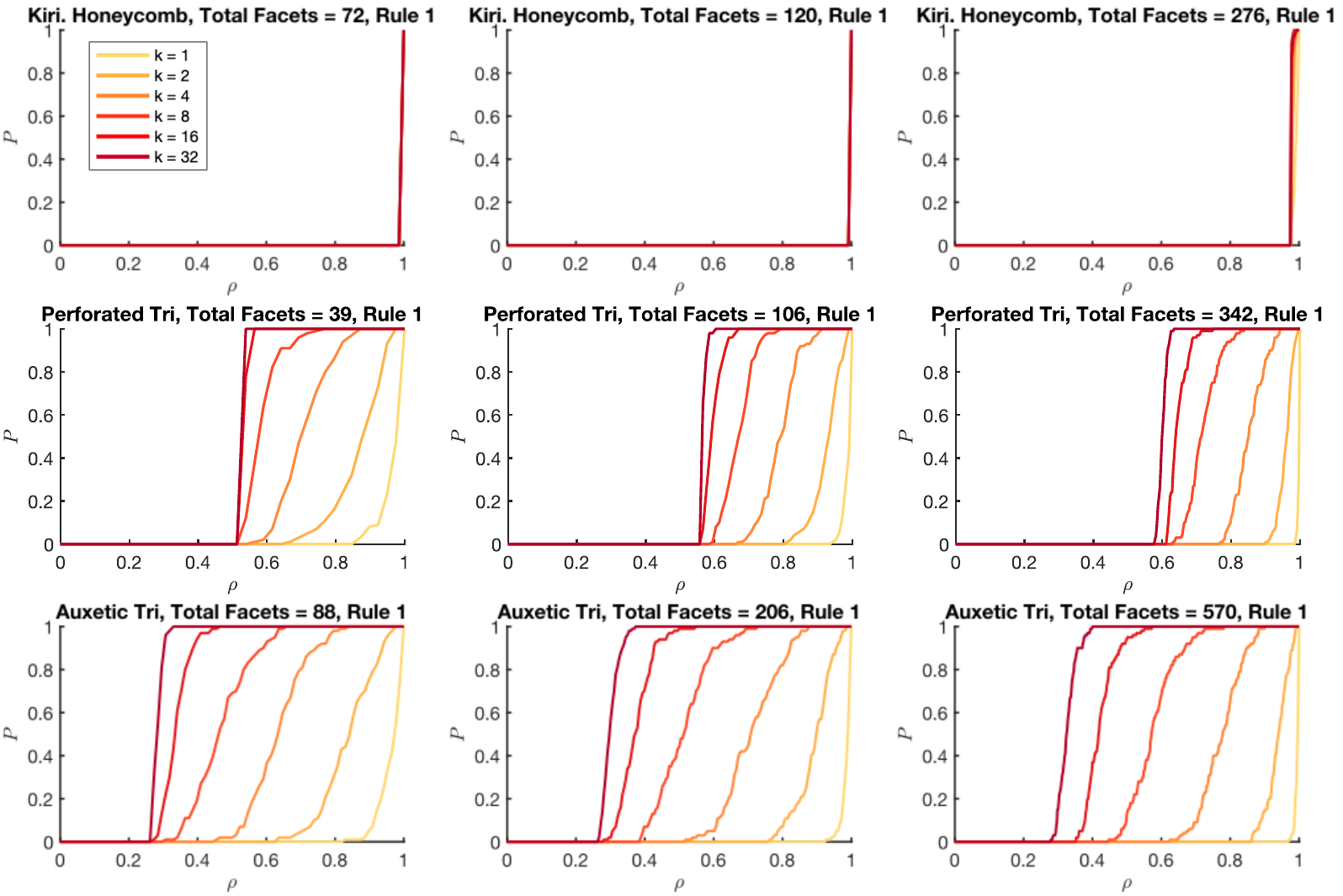}
    \caption{\textbf{Rigidity percolation in perforated origami under the Most Efficient selection rule for three origami structures with different sizes.} For different problem sizes (with the total number of facets indicated in each subfigure title) and different numbers of choices $k = 1, 2, 4, 8, 16, 32$, we compute the probability $P$ of obtaining a minimum-DOF structure at different planarity constraint densities $\rho$, based on 100 simulations.}
    \label{fig:SI_P_rule1_fig3}
\end{figure*}

\begin{figure*}[t!]
    \centering  
    \includegraphics[width=0.8\linewidth]{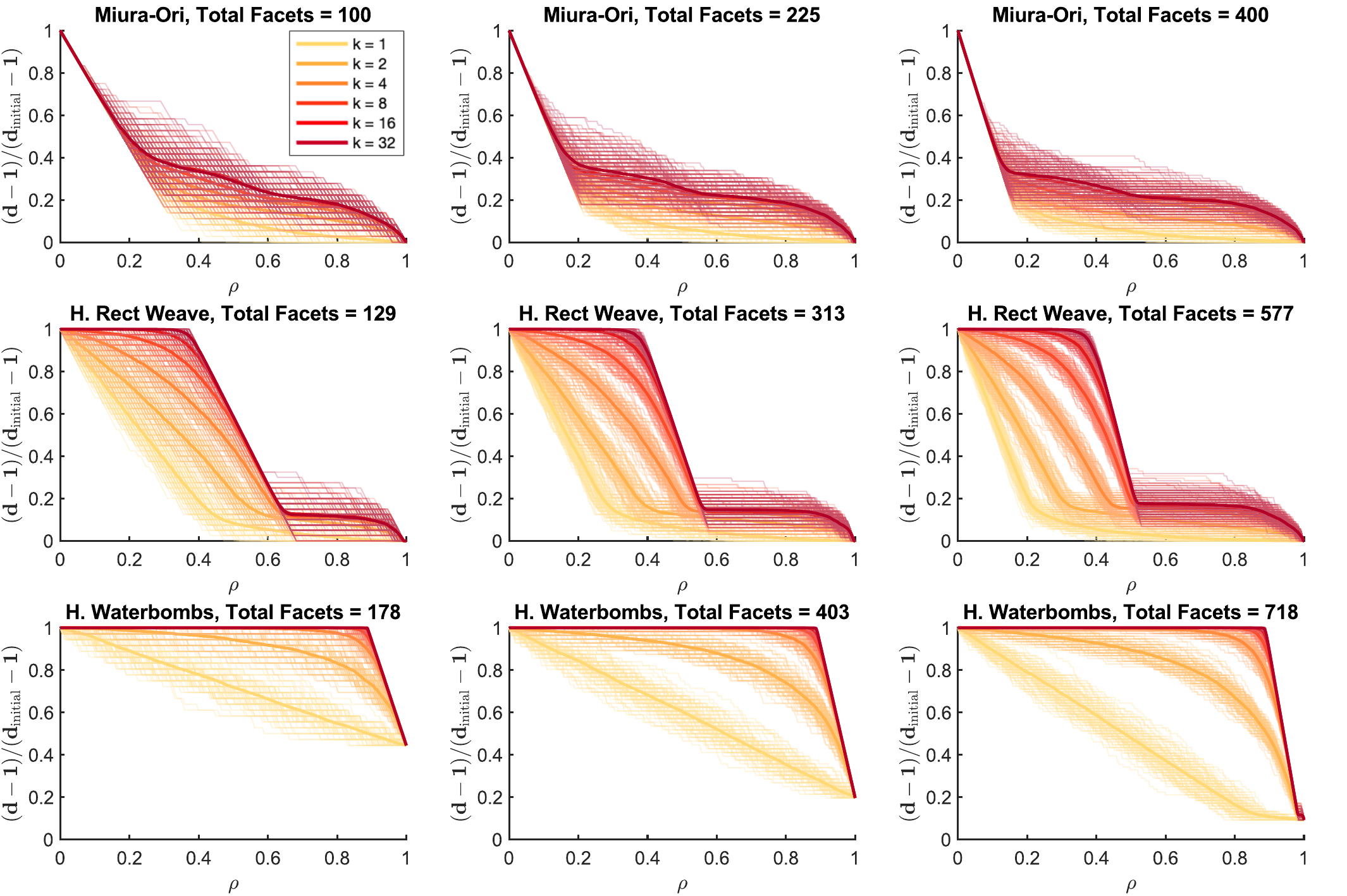}
    \caption{\textbf{Change in the normalized DOF under the Least Efficient selection rule with different numbers of choices for three types of periodic origami structures with different sizes.} For each type and each size, different numbers of choices $k = 1, 2, 4, 8, 16, 32$ are considered.}
    \label{fig:SIrule2_fig1}
\end{figure*}

\begin{figure*}[t!]
    \centering  
    \includegraphics[width=0.8\linewidth]{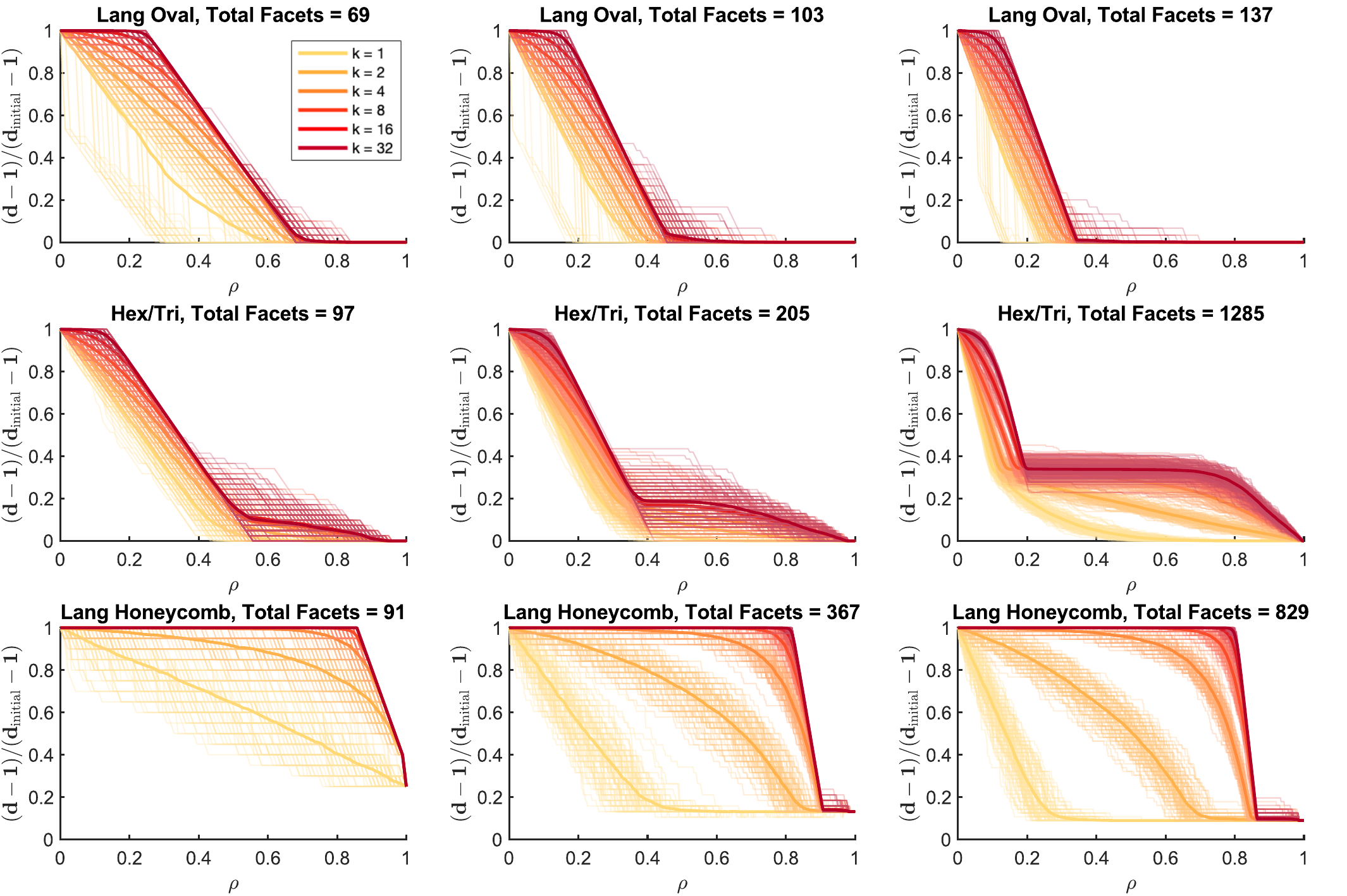}
    \caption{\textbf{Change in the normalized DOF under the Least Efficient selection rule with different numbers of choices for three types of rotational origami structures with different sizes.} For each type and each size, different numbers of choices $k = 1, 2, 4, 8, 16, 32$ are considered.}
    \label{fig:SIrule2_fig2}
\end{figure*}

\begin{figure*}[t!]
    \centering  
    \includegraphics[width=0.8\linewidth]{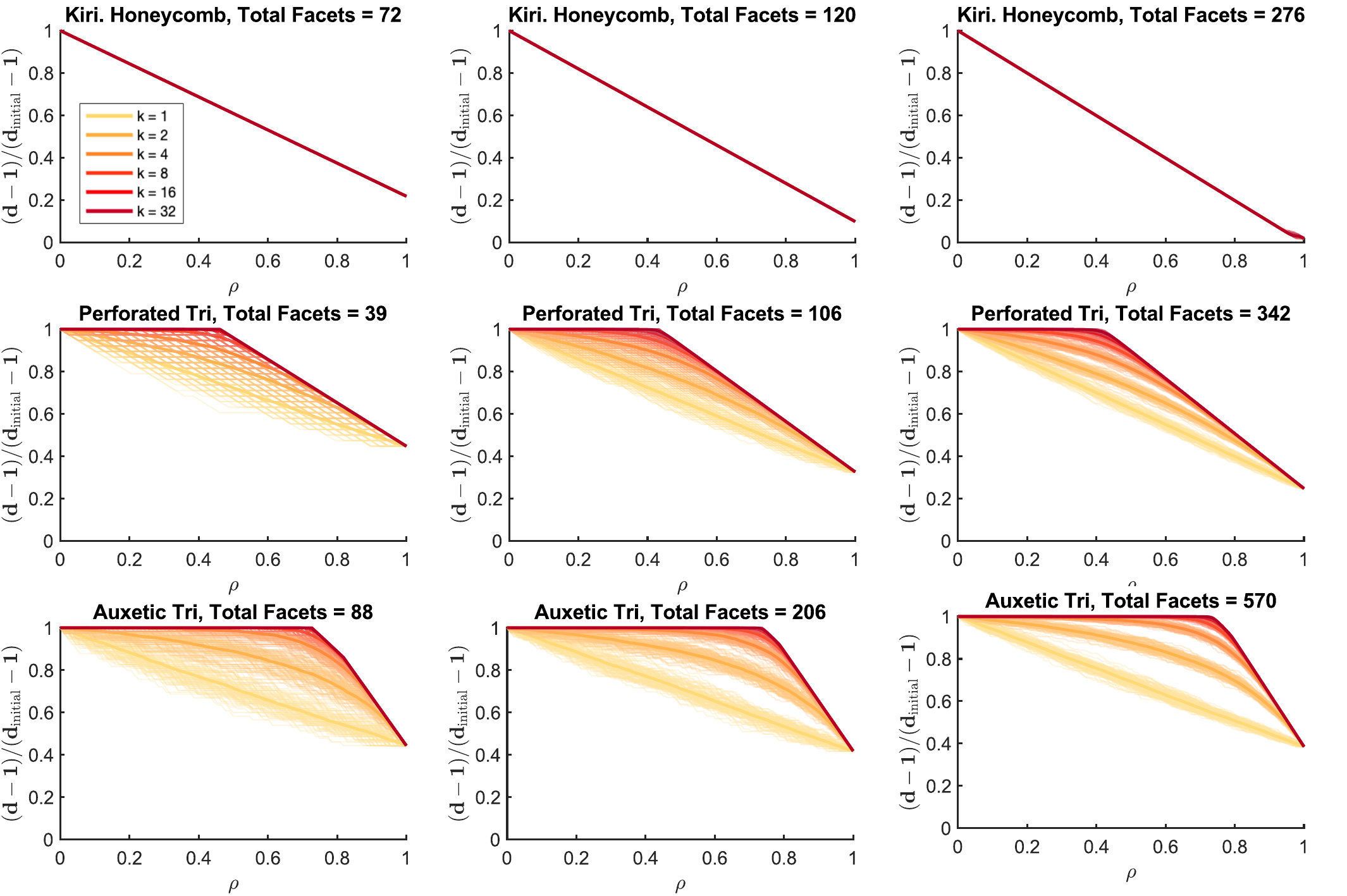}
    \caption{\textbf{Change in the normalized DOF under the Least Efficient selection rule with different numbers of choices for three types of perforated origami structures with different sizes.} For each type and each size, different numbers of choices $k = 1, 2, 4, 8, 16, 32$ are considered.}
    \label{fig:SIrule2_fig3}
\end{figure*}

\begin{figure*}[t!]
    \centering  
    \includegraphics[width=0.8\linewidth]{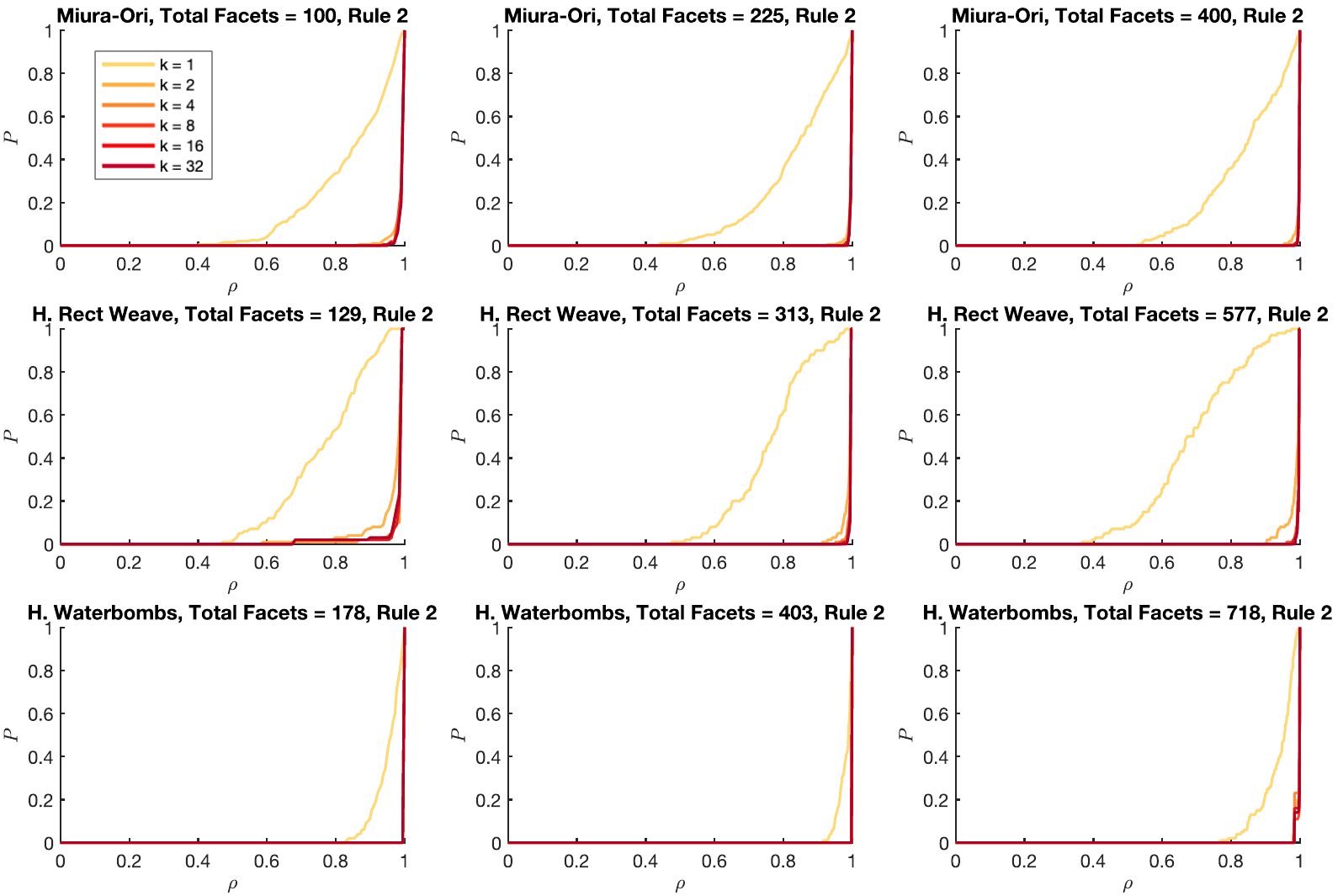}
    \caption{\textbf{Rigidity percolation in periodic origami under the Least Efficient selection rule for three origami structures with different sizes.} For different problem sizes (with the total number of facets indicated in each subfigure title) and different numbers of choices $k = 1, 2, 4, 8, 16, 32$, we compute the probability $P$ of obtaining a minimum-DOF structure at different planarity constraint densities $\rho$, based on 100 simulations.}
    \label{fig:SI_P_rule2_fig1}
\end{figure*}

\begin{figure*}[t!]
    \centering  
    \includegraphics[width=0.8\linewidth]{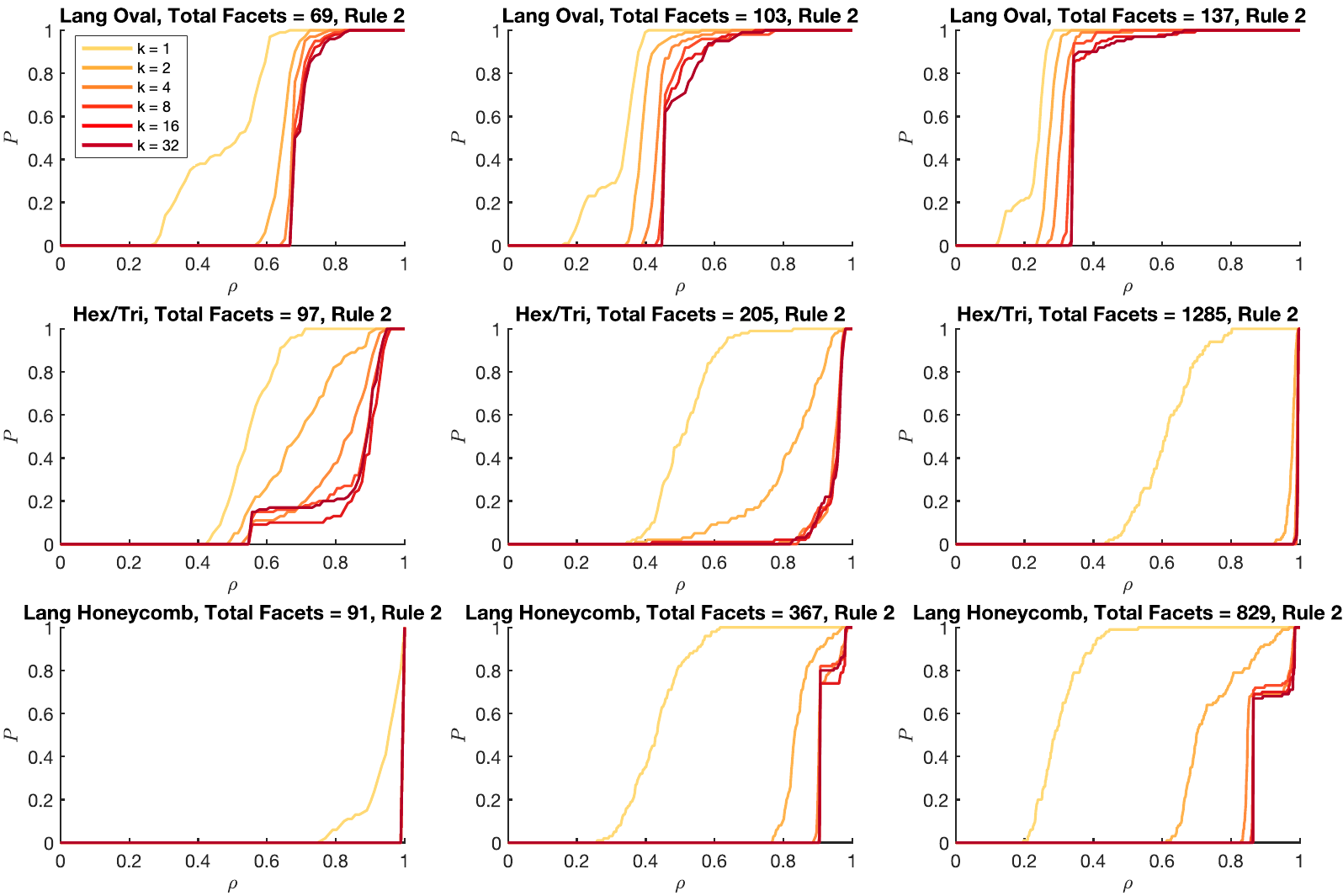}
    \caption{\textbf{Rigidity percolation in rotational origami under the Least Efficient selection rule for three origami structures with different sizes.} For different problem sizes (with the total number of facets indicated in each subfigure title) and different numbers of choices $k = 1, 2, 4, 8, 16, 32$, we compute the probability $P$ of obtaining a minimum-DOF structure at different planarity constraint densities $\rho$, based on 100 simulations.}
    \label{fig:SI_P_rule2_fig2}
\end{figure*}

\begin{figure*}[t!]
    \centering  
    \includegraphics[width=0.8\linewidth]{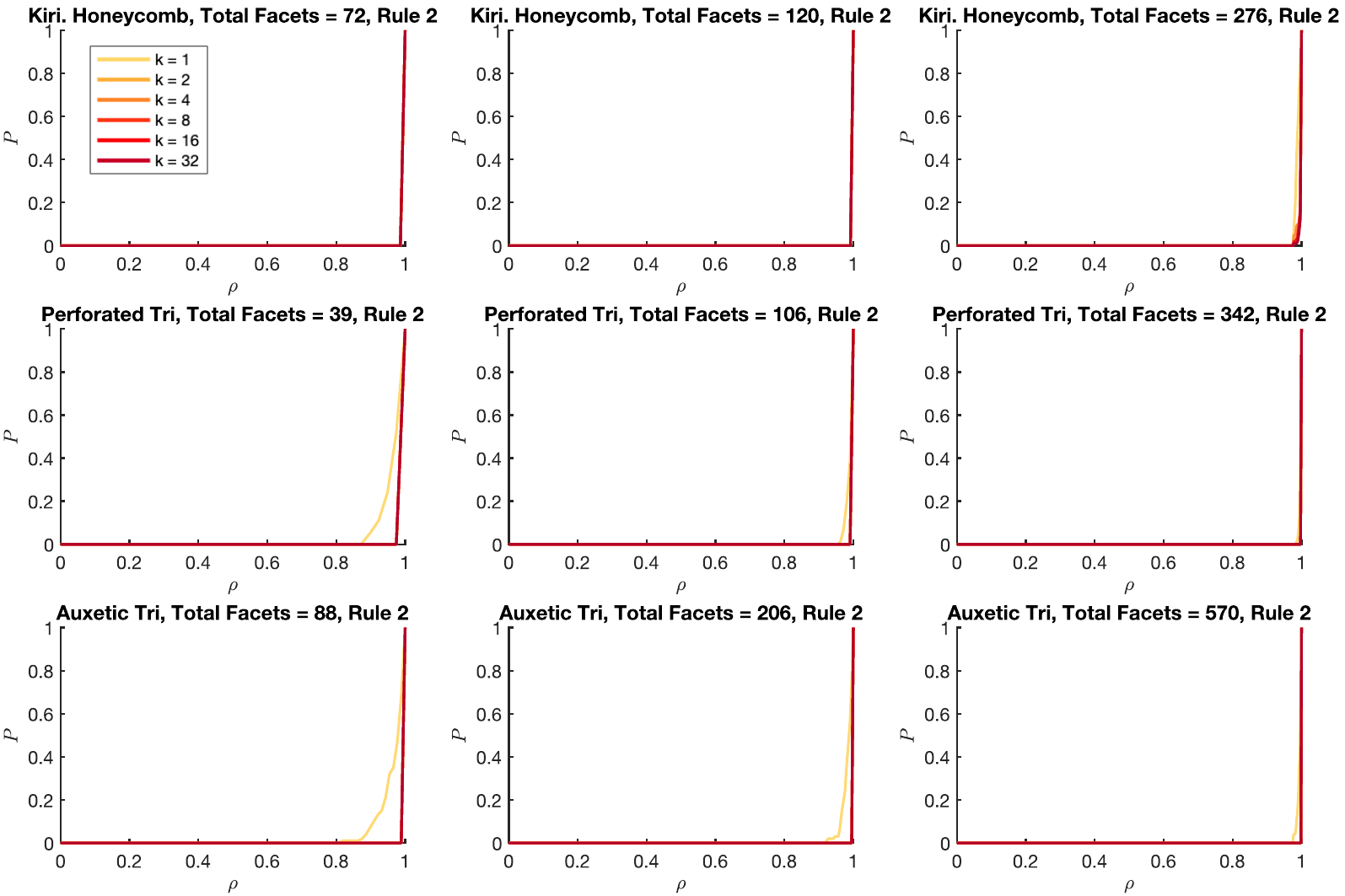}
    \caption{\textbf{Rigidity percolation in perforated origami under the Least Efficient selection rule for three origami structures with different resolutions.} For different problem sizes (with the total number of facets indicated in each subfigure title) and different numbers of choices $k = 1, 2, 4, 8, 16, 32$, we compute the probability $P$ of obtaining a minimum-DOF structure at different planarity constraint densities $\rho$, based on 100 simulations.}
    \label{fig:SI_P_rule2_fig3}
\end{figure*}

\begin{figure*}[t!]
    \centering  
    \includegraphics[width=\linewidth]{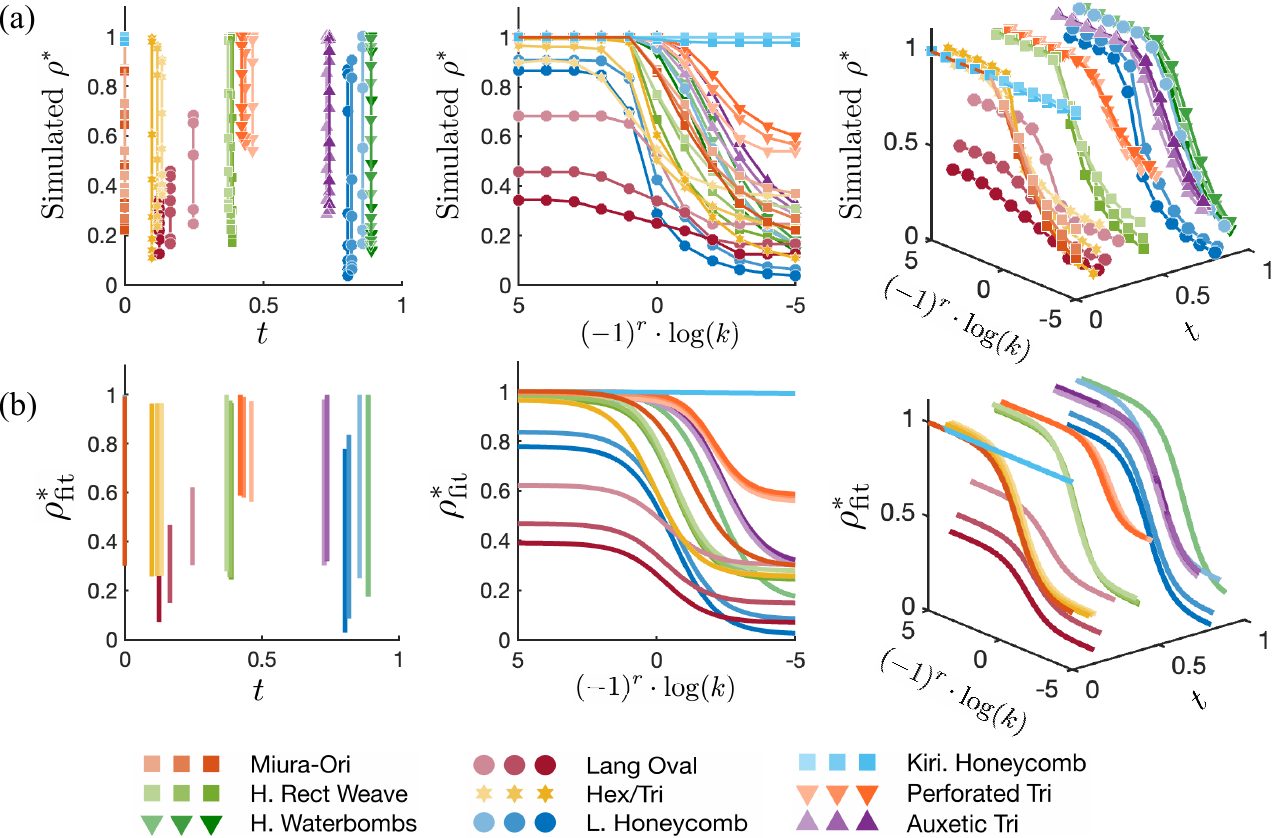}
    \caption{\textbf{The simulated critical transition density $\rho^{*}$ and the fitted values $\rho_{\text{fit}}^{*}$ for the nine different types of origami structures under different selection rules and different number of choices $k$.} The nine different types of origami structures are represented using different color and marker styles. For each structure, three resolutions are considered, indicated by different marker transparencies (lower transparency corresponds to lower resolution; see the caption of Fig.~\ref{fig:SI_sizes} for detailed resolution information).(a) Simulated results: critical transition density $\rho^{*}$ versus selection rule; $\rho^{*}$ versus triangular facet ratio; and a 3D plot combining both variables.
    (b) Fitted results: fitted critical transition density $\rho_{\text{fit}}^{*}$ versus selection rule; $\rho_{\text{fit}}^{*}$ versus triangular facet ratio; and a 3D plot combining both variables.
}
    \label{fig:SI_critical_rho_3d}
\end{figure*}

\begin{table*}[t]
    \centering
    \begin{tabular}{|c|C{25mm}|C{25mm}|C{25mm}|C{30mm}|C{30mm}|}
        \hline
        \textbf{Pattern Name} & \textbf{Number of Facets} & \textbf{Triangular Facet Ratio $t$} & \textbf{Number of Choices $k$} & \textbf{Most Efficient Selection Rule $\rho^*$} & \textbf{Least Efficient Selection Rule $\rho^*$} \\
        \hline
        \multirow{6}{*}{Miura-ori} & \multirow{6}{*}{100} & \multirow{6}{*}{0.00}
        & 1  & 0.8700 & 0.8700 \\
        &    &    & 2  & 0.7300 & 1.0000 \\
        &    &    & 4  & 0.6400 & 1.0000 \\
        &    &    & 8  & 0.4400 & 1.0000 \\
        &    &    & 16 & 0.3900 & 1.0000 \\
        &    &    & 32 & 0.3700 & 1.0000 \\
        \hline
        \multirow{6}{*}{Miura-ori} & \multirow{6}{*}{225} & \multirow{6}{*}{0.00}
        & 1  & 0.8711 & 0.8533 \\
        &    &    & 2  & 0.6844 & 1.0000 \\
        &    &    & 4  & 0.4844 & 1.0000 \\
        &    &    & 8  & 0.3822 & 1.0000 \\
        &    &    & 16 & 0.3111 & 1.0000 \\
        &    &    & 32 & 0.2667 & 1.0000 \\
        \hline
        \multirow{6}{*}{Miura-ori} &
        \multirow{6}{*}{400} & \multirow{6}{*}{0.00}
        & 1  & 0.8625 & 0.8550 \\
        &    &    & 2  & 0.6450 & 1.0000 \\
        &    &    & 4  & 0.4725 & 1.0000 \\
        &    &    & 8  & 0.3275 & 1.0000\\
        &    &    & 16 & 0.2550 & 1.0000 \\
        &    &    & 32 & 0.2200 & 1.0000\\
        \hline


        \multirow{6}{*}{\shortstack{Huffman\\Rectangular Weave}} & \multirow{6}{*}{129} & \multirow{6}{*}{0.37}
        & 1  & 0.7519 & 0.7909 \\
        &    &    & 2  & 0.5891 & 0.9845 \\
        &    &    & 4  & 0.4574 & 0.9922 \\
        &    &    & 8  & 0.3721 & 0.9922 \\
        &    &    & 16 & 0.3256 & 0.9922 \\
        &    &    & 32 & 0.3101 & 0.9922 \\
        \hline
        \multirow{6}{*}{\shortstack{Huffman\\Rectangular Weave}} &
        \multirow{6}{*}{313} & \multirow{6}{*}{0.38}
        & 1  & 0.7636 & 0.7688 \\
        &    &    & 2  & 0.5559 & 0.9936 \\
        &    &    & 4  & 0.3770 & 0.9968 \\
        &    &    & 8  & 0.2875 & 0.9968 \\
        &    &    & 16 & 0.2460 & 0.9968 \\
        &    &    & 32 & 0.2236 & 0.9968 \\
        \hline
        \multirow{6}{*}{\shortstack{Huffman\\Rectangular Weave}} &
        \multirow{6}{*}{577} & \multirow{6}{*}{0.39}
        & 1  & 0.6655 & 0.6724\\
        &    &    & 2  & 0.4333 & 0.9965 \\
        &    &    & 4  & 0.3016& 0.9983 \\
        &    &    & 8  & 0.2236& 0.9983\\
        &    &    & 16 & 0.1924& 0.9983\\
        &    &    & 32 & 0.1681 & 0.9983 \\
        \hline


        \multirow{6}{*}{Huffman Waterbombs} & 
        \multirow{6}{*}{178} & \multirow{6}{*}{0.89}
        & 1  & 0.9607 & 0.9607 \\
        &    &    & 2  & 0.8146 & 1.0000\\
        &    &    & 4  & 0.5843 & 1.0000 \\
        &    &    & 8  & 0.3652 & 1.0000 \\
        &    &    & 16 & 0.2360 & 1.0000 \\
        &    &    & 32 & 0.1629 & 1.0000 \\
        \hline
        \multirow{6}{*}{Huffman Waterbombs} &
        \multirow{6}{*}{403} & \multirow{6}{*}{0.89}
        & 1  & 0.9876 & 0.9901 \\
        &    &    & 2  & 0.8834 & 1.0000 \\
        &    &    & 4  & 0.6650 & 1.0000 \\
        &    &    & 8  & 0.4392 & 1.0000 \\
        &    &    & 16 & 0.2705& 1.0000 \\
        &    &    & 32 & 0.1787 & 1.0000 \\
        \hline
        \multirow{6}{*}{Huffman Waterbombs} & \multirow{6}{*}{718} & \multirow{6}{*}{0.89}
        & 1  & 0.9248 & 0.9554 \\
        &    &    & 2  & 0.7437 & 1.0000\\
        &    &    & 4  & 0.4875 & 1.0000\\
        &    &    & 8  & 0.3106 & 1.0000\\
        &    &    & 16 & 0.2033 & 1.0000\\
        &    &    & 32 & 0.1379 & 1.0000 \\
        \hline

\end{tabular}
    \caption{\textbf{The critical transition density~$\rho^*$ for different periodic origami structures under the Most Efficient and Least Efficient selection rules with different number of choices $k$}.}
    \label{tab:critical_density_1}
\end{table*}


\begin{table*}[t]
    \centering
    \begin{tabular}{|c|C{25mm}|C{25mm}|C{25mm}|C{30mm}|C{30mm}|}
        \hline
        \textbf{Pattern Name} & \textbf{Number of Facets} & \textbf{Triangular Facet Ratio $t$} & \textbf{Number of Choices $k$} & \textbf{Most Efficient Selection Rule $\rho^*$} & \textbf{Least Efficient Selection Rule $\rho^*$} \\
        \hline
        \multirow{6}{*}{Lang Oval} & \multirow{6}{*}{69} & \multirow{6}{*}{0.25}
        & 1  & 0.5217 & 0.5217 \\
        &    &    & 2  & 0.3043 & 0.6522 \\
        &    &    & 4  & 0.2464 & 0.6812 \\
        &    &    & 8  & 0.2464 & 0.6812 \\
        &    &    & 16 & 0.2464 & 0.6812 \\
        &    &    & 32 & 0.2464 & 0.6812 \\
        \hline
        \multirow{6}{*}{Lang Oval} &
        \multirow{6}{*}{103} & \multirow{6}{*}{0.17}
        & 1  & 0.3301 & 0.3495\\
        &    &    & 2  & 0.2913 & 0.3883 \\
        &    &    & 4  & 0.1845 & 0.4369 \\
        &    &    & 8  & 0.1650 & 0.4563 \\
        &    &    & 16 & 0.1650& 0.4563 \\
        &    &    & 32 & 0.1650 & 0.4563 \\
        \hline
        \multirow{6}{*}{Lang Oval} & \multirow{6}{*}{137} & \multirow{6}{*}{0.17}
        & 1  & 0.2482 & 0.2482\\
        &    &    & 2  & 0.2190 & 0.2774 \\
        &    &    & 4  & 0.1825 & 0.3066 \\
        &    &    & 8  & 0.1241 & 0.3358\\
        &    &    & 16 & 0.1241 & 0.3431\\
        &    &    & 32 & 0.1241 & 0.3431 \\
        \hline


        \multirow{6}{*}{Hex/Tri} & \multirow{6}{*}{97} & \multirow{6}{*}{0.13}
        & 1  & 0.5464 & 0.5464 \\
        &    &    & 2  & 0.4433 & 0.6907 \\
        &    &    & 4  & 0.4021 & 0.8351\\
        &    &    & 8  & 0.3711 & 0.8966 \\
        &    &    & 16 & 0.3608 & 0.9072 \\
        &    &    & 32 & 0.3608 & 0.8969 \\
        \hline
        \multirow{6}{*}{Hex/Tri} &
        \multirow{6}{*}{205} & \multirow{6}{*}{0.12}
        & 1  & 0.4976 & 0.5073 \\
        &    &    & 2  & 0.3951 & 0.8439 \\
        &    &    & 4  & 0.3171 & 0.9512 \\
        &    &    & 8  & 0.2732 & 0.9610 \\
        &    &    & 16 & 0.2439 & 0.9610 \\
        &    &    & 32 & 0.2341 & 0.9659 \\
        \hline
        \multirow{6}{*}{Hex/Tri} & \multirow{6}{*}{1285} & \multirow{6}{*}{0.10}
        & 1  & 0.6008 & 0.6086\\
        &    &    & 2  & 0.4233 & 0.9798 \\
        &    &    & 4  & 0.2739 & 0.9922 \\
        &    &    & 8  & 0.1852 & 0.9946\\
        &    &    & 16 & 0.1377 & 0.9946\\
        &    &    & 32 & 0.1097 & 0.9946 \\
        \hline


        \multirow{6}{*}{Lang Honeycomb} & \multirow{6}{*}{91} & \multirow{6}{*}{0.86}
        & 1  & 0.9560 & 0.9560\\
        &    &    & 2  & 0.7912 & 1.0000 \\
        &    &    & 4  & 0.5495 & 1.0000 \\
        &    &    & 8  & 0.3623 & 1.0000 \\
        &    &    & 16 & 0.2198 & 1.0000 \\
        &    &    & 32 & 0.1648 & 1.0000 \\
        \hline
        \multirow{6}{*}{Lang Honeycomb} &
        \multirow{6}{*}{367} & \multirow{6}{*}{0.82}
        & 1  & 0.4196 & 0.4305 \\
        &    &    & 2  & 0.2752 & 0.8338 \\
        &    &    & 4  & 0.1608 & 0.9046 \\
        &    &    & 8  & 0.1008 & 0.9074 \\
        &    &    & 16 & 0.0736& 0.9074 \\
        &    &    & 32 & 0.0627 & 0.9074 \\
        \hline
        \multirow{6}{*}{Lang Honeycomb} & \multirow{6}{*}{829} & \multirow{6}{*}{0.80}
        & 1  & 0.2871 & 0.2871 \\
        &    &    & 2  & 0.1580 & 0.6996 \\
        &    &    & 4  & 0.0965 & 0.8480 \\
        &    &    & 8  & 0.0627 & 0.8649\\
        &    &    & 16 & 0.0458 & 0.8649\\
        &    &    & 32 & 0.0374 & 0.8649\\
        \hline

    \end{tabular}
    \caption{\textbf{The Critical transition density $\rho^*$ of different rotational origami structures under the Most Efficient and Least Efficient selection rules with different number of choices $k$}.}
    \label{tab:critical_density_2}
\end{table*}

\begin{table*}[t]
    \centering
    \begin{tabular}{|c|C{25mm}|C{25mm}|C{25mm}|C{30mm}|C{30mm}|}
        \hline
        \textbf{Pattern Name} & \textbf{Number of Facets} & \textbf{Triangular Facet Ratio $t$} & \textbf{Number of Choices $k$} & \textbf{Most Efficient Selection Rule $\rho^*$} & \textbf{Least Efficient Selection Rule $\rho^*$} \\
        \hline
        
        \multirow{6}{*}{Kirigami Honeycomb} & \multirow{6}{*}{72} & \multirow{6}{*}{0.00}
        & 1  & 1.0000 & 1.0000 \\
        &    &    & 2  & 1.0000 & 1.0000  \\
        &    &    & 4  & 1.0000  & 1.0000  \\
        &    &    & 8  & 1.0000  & 1.0000  \\
        &    &    & 16 & 1.0000  & 1.0000  \\
        &    &    & 32 & 1.0000  & 1.0000  \\
        \hline
        \multirow{6}{*}{Kirigami Honeycomb} &
        \multirow{6}{*}{120} & \multirow{6}{*}{0.00}
        & 1  & 1.0000 & 1.0000 \\
        &    &    & 2  & 1.0000 & 1.0000  \\
        &    &    & 4  & 1.0000 & 1.0000  \\
        &    &    & 8  & 1.0000 & 1.0000 \\
        &    &    & 16 & 1.0000 & 1.0000 \\
        &    &    & 32 & 1.0000 & 1.0000 \\
        \hline
        \multirow{6}{*}{Kirigami Honeycomb} & \multirow{6}{*}{276} & \multirow{6}{*}{0.00}
        & 1  & 0.8711 & 0.4200 \\
        &    &    & 2  & 1.0000 & 1.0000 \\
        &    &    & 4  & 1.0000 & 1.0000 \\
        &    &    & 8  & 1.0000 & 1.0000\\
        &    &    & 16 & 1.0000 & 1.0000\\
        &    &    & 32 & 1.0000 & 1.0000 \\
        \hline
        

        \multirow{6}{*}{Auxetic Triangle} & \multirow{6}{*}{88} & \multirow{6}{*}{0.73}
        & 1  & 0.9773 & 0.9886 \\
        &    &    & 2  & 0.8523 & 1.0000 \\
        &    &    & 4  & 0.6364 & 1.0000 \\
        &    &    & 8  & 0.4659 & 1.0000 \\
        &    &    & 16 & 0.3409 & 1.0000 \\
        &    &    & 32 & 0.2814 & 1.0000 \\
        \hline
        \multirow{6}{*}{Auxetic Triangle} &
        \multirow{6}{*}{206} & \multirow{6}{*}{0.74}
        & 1  & 0.9903 & 0.9903 \\
        &    &    & 2  & 0.8981 & 1.0000 \\
        &    &    & 4  & 0.7089 & 1.0000 \\
        &    &    & 8  & 0.5149 & 1.0000 \\
        &    &    & 16 & 0.3738 & 1.0000 \\
        &    &    & 32 & 0.3010 & 1.0000 \\
        \hline
        \multirow{6}{*}{Auxetic Triangle} & \multirow{6}{*}{570} & \multirow{6}{*}{0.74}
        & 1  & 0.9965 & 0.9965 \\
        &    &    & 2  & 0.9404 & 1.0000\\
        &    &    & 4  & 0.7789 & 1.0000 \\
        &    &    & 8  & 0.5719 & 1.0000\\
        &    &    & 16 & 0.4175 & 1.0000\\
        &    &    & 32 & 0.3246 & 1.0000 \\
        \hline
        

        \multirow{6}{*}{Perforated Triangle} & \multirow{6}{*}{49} & \multirow{6}{*}{0.46}
        & 1  & 1.0000 & 1.0000 \\
        &    &    & 2  & 0.8974 & 1.0000 \\
        &    &    & 4  & 0.7179 & 1.0000 \\
        &    &    & 8  & 0.5897 & 1.0000 \\
        &    &    & 16 & 0.5385 & 1.0000 \\
        &    &    & 32 & 0.5385 & 1.0000 \\
        \hline
        \multirow{6}{*}{Perforated Triangle} &
        \multirow{6}{*}{106} & \multirow{6}{*}{0.43}
        & 1  & 1.0000 & 1.0000 \\
        &    &    & 2  & 0.9340 & 1.0000\\
        &    &    & 4  & 0.7925 & 1.0000 \\
        &    &    & 8  & 0.6698 & 1.0000 \\
        &    &    & 16 & 0.5940 & 1.0000 \\
        &    &    & 32 & 0.5560 & 1.0000 \\
        \hline
        \multirow{6}{*}{Perforated Triangle} & \multirow{6}{*}{225} & \multirow{6}{*}{0.00}
        & 1  & 1.0000 & 1.0000 \\
        &    &    & 2  & 0.9620 & 1.0000 \\
        &    &    & 4  & 0.8509 & 1.0000\\
        &    &    & 8  & 0.7135& 1.0000\\
        &    &    & 16 & 0.6404 & 1.0000\\
        &    &    & 32 & 0.5994 & 1.0000 \\
        \hline

    \end{tabular}
    \caption{\textbf{The critical transition density $\rho^*$ of different perforated origami structures under the Most Efficient and Least Efficient selection rules with different number of choices $k$}.}
    \label{tab:critical_density_3}
\end{table*}

\begin{table*}[t]
\centering
\begin{tabular}{|l|r|r|r|r|r|r|}
\hline
\textbf{Structure} & $\mathbf{a}$ & $\mathbf{b}$ & $\mathbf{c}$ & $\mathbf{d}$ & $\mathbf{f}$ & \textbf{RMSE} \\
\hline
Miura-Ori           & 0.353314 & 0.615089 & 0.766633 & 0.000000  & 0.647008  & 0.035586 \\
Huffman Rectangular Weave       & 0.361662 & 0.727425 & 0.475935 & -2.124972 & 1.429219  & 0.047417 \\
Huffman Waterbombs       & 0.423211 & 0.617756 & 1.279683 & -1.167508 & 1.613251  & 0.036420 \\
Lang Oval           & 0.160092 & 0.632491 & 0.212309 & 1.900514  & -0.005942 & 0.051558 \\
Hex/Tri             & 0.354397 & 0.655452 & -0.038231 & 0.048151  & 0.605736  & 0.072726 \\
Lang Honeycomb        & 0.377366 & 0.629517 & 0.382077 & 4.134894  & -2.920914 & 0.104569 \\
Kirigami Honeycomb     & 2.132125 & 0.000368 & 0.115862 & 0.000000  & 0.750164  & 0.006987 \\
Perforated Triangle      & 0.208261 & 0.831877 & 1.741861 & -0.650196 & 1.065508  & 0.025120 \\
Auxetic Triangle         & 0.349825 & 0.656863 & 1.559848 & 1.884478  & -0.740271 & 0.026843 \\
\hline
\end{tabular}

\caption{\textbf{The fitted model parameters and the root mean square error (RMSE) for each type of origami structure considered in this work.}}
\label{tab:fit_params}
\end{table*}

\end{document}